%% file: main.tex
\documentclass[journal]{IEEEtran}
\usepackage{graphicx}
\usepackage{overpic}
\usepackage{bbm}
\usepackage{braket}
\usepackage{booktabs}
\usepackage[sort]{cite}
\usepackage[linesnumbered,ruled,vlined]{algorithm2e} 
\usepackage[most]{tcolorbox}
\usepackage{amsthm}

\usepackage{lmodern}
\usepackage{scalerel}

\input{./macro/macros}
\input{./macro/local-macros}

\newcommand{\change}[1]{#1}

\usepackage{url}
\usepackage[breaklinks,
  colorlinks,
  linkcolor={blue!50!black},
  citecolor={blue!50!black},
  urlcolor={blue!50!black}]{hyperref}
\usepackage[nameinlink,noabbrev,capitalise]{cleveref}
\usepackage{float}

\usepackage{pifont}
\newcommand{\cmark}{\ding{51}}%
\newcommand{\xmark}{\ding{55}}%

\newtheorem{assumption}{Assumption}
\newtheorem{proposition}{Proposition}
\newtheorem{definition}{Definition}
\newtheorem{corollary}{Corollary}
\newtheorem{theorem}{Theorem}
\newtheorem{lemma}{Lemma}
\crefname{assumption}{Assumption}{Assumptions}
\crefname{subsection}{section}{sections}
\crefname{paragraph}{Paragraph}{Paragraphs}

\usepackage{enumitem}
\setlist{nolistsep}
\newlist{propenum}{enumerate}{1} 
\setlist[propenum]{label=\alph*{\rm)}, ref=\theproposition(\alph*)}
\crefalias{propenumi}{proposition}
\newlist{corenum}{enumerate}{1} 
\setlist[corenum]{label=\alph*{\rm)}, ref=\thecorollary(\alph*)}
\crefalias{corenumi}{corollary}
\newlist{thmenum}{enumerate}{1} 
\setlist[thmenum]{label=\alph*{\rm)}, ref=\thetheorem(\alph*)}
\crefalias{thmenumi}{theorem}
\newlist{lemmaenum}{enumerate}{1} 
\setlist[lemmaenum]{label=\alph*{\rm)}, ref=\thelemma(\alph*)}
\crefalias{lemmaenumi}{lemma}

\begin{document}
%
\title{Polar Deconvolution of Mixed Signals}
%
%
%

\author{Zhenan Fan\thanks{Department of Computer Science, University of British Columbia, Vancouver, BC, Canada}\footnotemark[1]
, Halyun Jeong\thanks{Department of Mathematics, University of California, Los Angeles, California, United States}\footnotemark[2]
, Babhru Joshi\thanks{Department of Mathematics, University of British Columbia, Vancouver, BC, Canada}\footnotemark[3]
, Michael P. Friedlander\footnotemark[1, 3]}
\maketitle

\begin{abstract}
The signal demixing problem seeks to separate a superposition of multiple signals into its constituent components. This paper studies a two-stage approach that first decompresses and subsequently deconvolves the noisy and undersampled observations of the superposition using two convex programs. Probabilistic error bounds are given on the accuracy with which this process approximates the individual signals. The theory of polar convolution of convex sets and gauge functions plays a central role in the analysis and solution process.  If the measurements are random and the noise is bounded, this approach stably recovers low-complexity and mutually incoherent signals, with high probability and with near optimal sample complexity. We develop an efficient algorithm, based on level-set and conditional-gradient methods, that solves the convex optimization problems with sublinear iteration complexity and linear space requirements. Numerical experiments on both real and synthetic data confirm the theory and the efficiency of the approach.
\end{abstract}


%
\IEEEpeerreviewmaketitle

\section{Introduction}\label{sec:1}

The signal demixing problem seeks to separate a superposition of signals into its constituent components. In the measurement model we consider, a set of signals $\{x\nai\}_{i=1}^k$ in $\Re^n$ are observed through noisy
measurements $b\in\Re^m$, with $m\le n$, of the form
\begin{equation}\label{eq:main-problem}
  b = M\xagg + \eta \text{with} \xagg\coloneqq \sum\limits_{i = 1}^k x\nai.
\end{equation}
The known linear operator $M:\Re^n \rightarrow \Re^m$ models the acquisition process of the superposition vector $\xagg$. The vector $\eta\in \Re^m$ represents noise uncorrelated with the data. This measurement model and its variations are useful for a range of data-science applications, including mixture models~\cite{araki2009blind,quiros2012dependent}, blind deconvolution~\cite{ahmed2013blind}, blind source separation~\cite{chan2008convex}, and morphological component analysis~\cite{bobin2007morphological}.

A central concern of the demixing problem \eqref{eq:main-problem} is to delineate efficient procedures and accompanying conditions that make it possible to recover the constituent signals to within a prescribed accuracy---using the fewest number of measurements $m$. The recovery of these constituent signals cannot be accomplished without additional information, such as the latent structure in each signal $x\nai$. We build on the general atomic-sparsity framework formalized by Chandrasekaran et al.~\cite{chandrasekaran2012convex}, and assume that each signal $x\nai$ is itself well represented as a superposition of a few atomic signals from a collection  $\Ascr_i\subset\Re^n$. In other words, the vectors $\{x\nai\}_{i=1}^k$ are paired with atomic sets $\{\Ascr_i\}_{i=1}^k$ that allow for \change{nonnegative decompositions} of the form
\begin{equation} \label{eq:decomposition}
  x\nai = \sum_{a\in \Ascr_i}c_a a, \quad c_a\geq 0,\quad \forall a \in \Ascr_i,
\end{equation}
where most of the coefficients $c_a$ are zero. This model of atomic sparsity includes a range of important notions of sparsity, such as sparse vectors, which are sparse in the set of canonical vectors, and low-rank matrices, which are sparse in the set of rank-1 matrices with unit spectral norm. Other important generalizations include higher-order tensor decompositions, useful in computer vision~\cite{vasilescu2002multilinear} and handwritten digit classification~\cite{savas2007digit}, and polynomial atomic decomposition~\cite{carando2008atomic}.

\change{The nonnegative constraints on the coefficients $c_a$ in~\eqref{eq:decomposition} are key for the validity of this decomposition, particularly for the case where the atomic sets $\Ascr_i$ are not \emph{centrosymmetric}---i.e., when $a\in\Ascr_i$ does not imply that ${-}{a}\in\Ascr_i$. For example, suppose that the ground-truth signal $x\nai$ is known to be low rank and positive definite. In that case, we would choose the corresponding atomic set
\[
  \Ascr_i=\set{uu^T \mid \|u\|=1},
\]
which is the set of symmetric unit-norm rank-1 matrices. The the nonnegativity constraint in~\eqref{eq:decomposition} then ensures that all nontrivial combinations of the atoms in $\Ascr_i$ produce positive definite matrices. The nonnegative constraint on the coefficients $c_a$ can be dropped only in the special case where the atomic set $\Ascr_i$ is centrosymmetric, but this is not an assumption that we make.%
}

A common approach to recover an atomic signal is to use the gauge function 
\begin{equation*}
  \gauge\As(x) \coloneqq \inf_{c_a}\left\{\sum_{a \in \Ascr} c_a \mid x = \sum_{a\in \Ascr}c_a a,\enspace c_a \geq 0\enspace \forall a \in \Ascr\right\},
\end{equation*}
where $\Ascr$ is the atomic set for $x$. This gauge function is central to the formulation of convex optimization process that provably leads to solutions that have sparse decompositions in the sense of~\eqref{eq:decomposition}. The properties of gauges and their relationship with atomic sparsity have been well-studied in the literature and are outlined in Chandrasekaran et al.~\cite{chandrasekaran2012convex} and Fan et al.~\cite{fan2019alignment}.

The typical approach to the demixing problem is to combine $k$ separate gauge functions, each corresponding to one of the atomic sets $\{\Ascr_i\}_{i=1}^k$, as a weighted sum or similar formulations. We instead combine the $k$ separate gauge functions using a special-purpose convolution operation called polar convolution, that can reflect the additive structure of the superposition, as defined in~\eqref{eq:main-problem}.

\subsection{Polar convolution} \label{sec-polar-deconvolution-into-components}

For any two atomic sets $\Ascr_1$ and $\Ascr_2$, the polar convolution of the corresponding gauge functions is 
\[
  (\gauge\Aso\maxconv\gauge\Ast)(x)
  \coloneqq  \inf_{x_1,\,x_2}\ \max_{x=x_1+x_2}\set{\gauge\Aso(x_1),\ \gauge\Ast(x_2)}.
\]
The resulting function is the gauge to the vector sum
\[
  \Ascr_1+\Ascr_2 = \set{a_1+a_2 | a_1\in\Ascr_1,\ a_2\in\Ascr_2},
\]
which is confirmed by the identity
\begin{equation}\label{eq-polar-convolution-sum}
  \gauge\Aso\maxconv\gauge\Ast = \gauge_{\scriptscriptstyle\Ascr_1+\Ascr_2},
\end{equation}
see Friedlander et al.~\cite[Proposition~6.2]{friedlander2019polarconvolution} and \cref{fig:sum-sets}.

\begin{figure}[t]
  \centering
  \includegraphics[width=.6\linewidth, page=1]{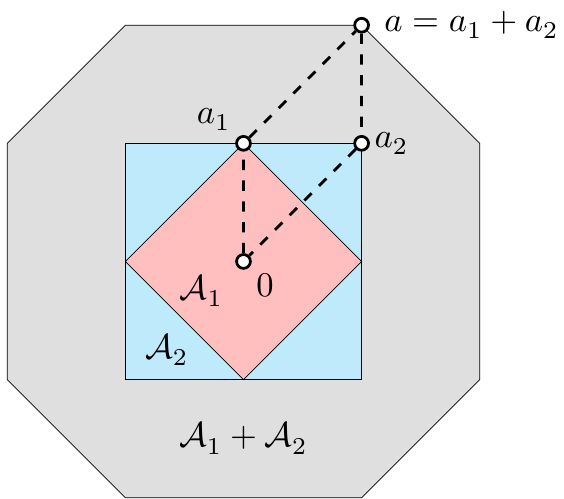}
  \caption{The sum of two atomic sets. The sum \(\Ascr_1+\Ascr_2\) is the unit level set for the polar convolution $\gauge\Aso\maxconv\gauge\Ast$, i.e., $\Ascr_1+\Ascr_2 = \set{a \mid \gauge\Aso\maxconv\gauge\Ast(a) \leq 1}$.\label{fig:sum-sets}}
\end{figure}

 The subdifferential properties of polar convolution facilitate our analysis and allow us to build an efficient algorithm that is practical for a range of problems. In particular, the polar convolution decouples under a duality correspondence built around the polarity of convex sets. The polar to a convex set $\Cscr\subset\Real^n$,
\[
  \Cscr\polar = \set{y \in\Re^n \mid \ip{x}{y}\leq 1 \mbox{ for all } x \in \Cscr},
\]
contains a dual description of $\Cscr$ in terms of all of its supporting hyperplanes. \change{Note that when $\Cscr$ is a convex cone, the constant $1$ in the above definition can be replaced with $0$.} Under this dual representation,
\[
  \gauge_{\scriptscriptstyle(\Ascr_1+\Ascr_2)\polar}
  = \gauge_{\scriptscriptstyle\Ascr_1\polar} + \gauge_{\scriptscriptstyle\Ascr_2\polar},
\]
which implies that the convex subdifferential decouples as 
$$\partial \gauge_{\scriptscriptstyle(\Ascr_1+\Ascr_2)\polar} = \partial\gauge_{\scriptscriptstyle\Ascr_1\polar}+\partial\gauge_{\scriptscriptstyle\Ascr_2\polar}.$$ Thus, a subgradient computation, which is central to all first-order methods for convex optimization, can be implemented using only subdifferential oracles for each of the polar functions $\gauge_{\scriptscriptstyle\Ascr_i\polar}$. We show in \cref{sec-algorithms} how to use this property to implement a version of the conditional gradient method~\cite{frank1956algorithm,jaggi2013revisiting} to obtain the polar decomposition using space complexity that scales linearly with the size of the data.

\subsection{Decompression and deconvolution}\label{sec:decompression-deconvolution}

The principle innovation of our approach to the demixing problem \eqref{eq:main-problem} is to decouple the recovery procedure into two stages: an initial \emph{decompression} stage meant to recover the superposition $\xagg$ from the vector of observations $b$, followed by a \emph{deconvolution} stage that separates the recovered superposition $\xagg$ into its constituent components $\{x\nai\}_{i=1}^k$. We couple the convex theory of polar convolution~\cite{friedlander2019polarconvolution} to the theory of statistical dimension and signal incoherence to derive a recovery procedure and analysis for demixing a compressively sampled mixture to within a prescribed accuracy.

\paragraph{Stage 1: Decompression} The initial decompression stage is based on the observation
that because each signal $x\nai$ is $\Ascr_i$ sparse, the superposition $\xagg$
must be sparse with respect to the weighted vector sum
\begin{equation}  \label{eq:weighted-atomic-sum}
\begin{aligned}
  \Ascr\subs &\coloneqq\sum_{i=1}^k\lambda_i\Ascr_i
  \\         &\equiv \Set{\sum_{i=1}^k\lambda_ia_i | a_i\in\Ascr_i\cup\set{0},\ i\in\irange1k}
\end{aligned}
\end{equation}
of the individual atomic sets $\Ascr_i$. The positive weights $\lambda_i$ carry information about the relative powers of the individual signals, and serve to equilibrate the gauge values of each signal. Thus, the weights $\lambda_i$ are defined so that for each $i\in\irange1k$,
\begin{equation}\label{eq-lambda-def}
  \gauge_{\scriptscriptstyle\lambda_i\Ascr_i}(x\nai)=\gauge\Aso(x\nao). 
\end{equation}
The initial decompression stage solves the convex optimization problem
\begin{equation} \tag{P1} \label{eq:decompression}
  \minimize{x\in\Re^n}\enspace \gauge\Ass(x) \enspace\st\enspace \|Mx - b\|_2 \leq \alpha,
\end{equation}
where the parameter $\alpha\ge0$ bounds the acceptable level of misfit between the linear model \(Mx\) and the observations $b$, and correspondingly reflects the anticipated magnitude of the noise $\eta$. It follows from~\eqref{eq-polar-convolution-sum} that the objective of \eqref{eq:decompression} is in fact the polar convolution of the individual weighted gauges:
\begin{equation*}
   \gauge_{\scriptscriptstyle\Ascr\subs}(x)
   = \gauge_{\scriptscriptstyle\lambda_1\Ascr_1}\maxconv\gauge_{\scriptscriptstyle\lambda_2\Ascr_2}\maxconv\cdots\ \maxconv\gauge_{\scriptscriptstyle\lambda_k\Ascr_k}(x).
\end{equation*}
\Cref{thm-stability} establishes conditions under which the solution $x\subs^*$ to~\eqref{eq:decompression} stably approximates the superposition~$\xagg$.

\paragraph{Stage 2: Deconvolution} The solution $x\subs^*$ of the decompression problem \eqref{eq:decompression} defines the subsequent convex deconvolution problem
\begin{equation} \tag{P2} \label{eq:deconvolution}
  \begin{array}{ll}
  \minimize{x_1, \ldots, x_k} &
  \max_{i\in1:k}\, \gauge_{\scriptscriptstyle\lambda_i\Ascr_i}(x_i)
\\\st & 
  \textstyle\sum_{i = 1}^k x_i = x\subs^*
  \end{array}
\end{equation}
to obtain approximations $x_i^*$ to each constituent signal $x\nai$. 

\begin{table*}[tb]
  \begin{center}
  \begin{tabular}{lc@{\quad}c@{\quad}c@{\quad}c@{\quad}c} 
  \toprule
            & Compressed   & Noisy        & Number& Recovery  & Explicit \\
  Reference & measurements & observations & of signals  & algorithm & error bound\\\midrule
  McCoy et al.~\cite{mccoy2014convexity} & \xmark & \xmark & $2$  &\cmark  &\xmark\\
  McCoy and Tropp~\cite{mccoy2014sharp} & \xmark & \xmark & $2$ & \xmark & \cmark \\
  Oymak and Tropp~\cite{oymak2017universality} & \cmark & \xmark & $2$ &\xmark & \cmark \\
  McCoy and Tropp~\cite{mccoy2013achievable} & \cmark & \cmark & $\geq 2$ &\xmark & \xmark \\
  This paper & \cmark & \cmark & $\geq 2$ & \cmark & \cmark \\\bottomrule 
  \end{tabular}
  \end{center}
  \caption{Comparison of the main mathematical results obtained by this paper and related references. Only this paper and McCoy and Tropp~\cite{mccoy2013achievable} consider the case of two or more signals.} \label{tab:comparasion}
\end{table*}

In both stages, a variant of the conditional-gradient method provides a
computationally and memory efficient algorithm that can be implemented with
storage proportional to the number of measurements $m$~\cite{fan2019alignment}. We
describe in \cref{sec-algorithms} the details of the method.

\subsection{Related work}\label{sec:1.2}

The history of signal demixing can be traced to early work in seismic imaging~\cite{claerbout1973robust} and morphological component analysis~\cite{starck2005morphological,bobin2007morphological}, which used 1-norm regularization to separate incoherent signals. More recently, McCoy and Tropp~\cite{mccoy2014sharp,mccoy2013achievable} and Oymak and Tropp~\cite{oymak2017universality} proposed a unified theoretical framework for signal demixing using modern tools from high-dimensional geometry.

McCoy et al.~\cite{mccoy2014convexity} analyzed the recovery guarantees of a convex program that can reconstruct $k = 2$ randomly-rotated signals from a full set of noiseless observations, i.e., $M$ is the identity matrix and $\norm{\eta}=0$. They also provided an ADMM-type algorithm for solving their proposed model. McCoy and Tropp~\cite{mccoy2014sharp} enhanced the error bound analysis under the same problem setting. McCoy and Tropp~\cite{mccoy2013achievable} subsequently extended this framework to demixing $k \geq 2$ randomly-rotated signals from noisy measurements, as modeled by~\eqref{eq:main-problem}. However, the constants in the recovery error bound are not made explicit. We postpone to \cref{sec:comparasion} a detailed comparison between our theoretical results and theirs. Oymak and Tropp~\cite{oymak2017universality} considered a demixing problem similar to~\eqref{eq:deconvolution} that also incorporates the measurement operator $M$, and provided guarantees for demixing two unknown vectors from random and noiseless measurements. We build on this line of work by providing explicit recovery error bounds in terms of the complexity of the signal sets and the number of measurements. Our analysis allows for any number of individual signals $k\ge2$. Moreover, we provide a memory-efficient algorithm for solving our proposed model. \Cref{tab:comparasion} compares main mathematical results obtained by this paper and the above references.

Early work on demixing sparse signals implicitly assumed some notion of incoherence between representations of the signals. This concept was made concrete by Donoho and Huo~\cite{doh01}, and subsequently Donoho and Elad~\cite{doe03}, who measured the mutual incoherence of finite bases via the maximal inner-products between elements of the sets. Related incoherence definitions appear in compressed sensing~\cite{tro04, maleki2009coherence} and robust PCA~\cite{candes2011robust, wright2013compressive}. In this paper we adopt McCoy and Tropp's~\cite{mccoy2013achievable} notion of incoherence as the minimal angle between conic representation of the individual signals.

\subsection{Roadmap and contributions}

\Cref{sec:2} shows that the decompression problem \eqref{eq:decompression} can stably recover the superposition $x\nag$. \Cref{thm:tropp} characterizes the recovery error in terms of the overall complexity of the signal, provided the measurements are random. This result follows directly from Tropp~\cite{tropp2015convex} and a conic decomposition property particular to polar convolution. \Cref{sec:3} shows that the deconvolution problem~\eqref{eq:deconvolution} can stably approximate each constituent signal $x\nai$. The bound in the recovery error is given in terms of the error in the initial decompression process and the incoherence between signals as measured by the minimum angle between conic representations of each signal; see \cref{thm-stability}. This result requires a general notion of incoherence based on the angle between descent cones, first analyzed by McCoy and Tropp~\cite{mccoy2013achievable}.  \Cref{sec-incoherence} shows how a random-rotation model yields a particular level of incoherence with high probability; see \cref{thm:Incoherence}. In that section we also develop the recovery guarantee under the random-rotation model; see \cref{corollary-stability}.  \Cref{sec-algorithms} outlines an algorithm based on conditional-gradient and level-set methods for computing the decompression and deconvolution process. The worst-case computational complexity of this process is sublinear in the required accuracy. \Cref{sec:6} describes numerical experiments on real and synthetic structured signals.

Proofs of all mathematical statements are given in Appendix~\ref{sec:7}.

In summary, the contributions of this paper are as follows.
\begin{itemize}
  \item We develop a two-stage, decompression-deconvolution approach for compressed signal demixing that is motivated by the polar convolution of gauge functions; see~\cref{sec:decompression-deconvolution}.
  \item Under the assumption of Gaussian measurements and randomly rotated signals, we develop explicit signal-recovery error and sample complexity bounds for the two-stage approach; see~\cref{corollary-stability}. These are the first known explicit error bounds for recovering an arbitrary number of mixed and compressed signals. 
  \item We propose an algorithm based on conditional-gradient and level-set method to solve our proposed model; see~\cref{alg-level-set} and~\cref{alg-vanilla-fw}. Our implementation is publicly available at \url{https://github.com/MPF-Optimization-Laboratory/AtomicOpt.jl}. 
  \item Extensive numerical experiments on synthetic and real data, described in \cref{sec:6}, verify the correctness of our theoretical results and the effectiveness of our approach.
\end{itemize}

\change{\subsection{Notation and assumption}

Throughout this paper, we use capital Roman letters $A, B,\ldots,$ to denote matrices or linear operators; lowercase Roman letters $a, b, \ldots,$ to denote vectors; calligraphic letters $\Ascr, \Bscr, \ldots,$ to denote sets; and lowercase Greek letters $\alpha, \beta, \ldots,$ to denote scalars. The 2-norm of a vector $z\in\Real^n$ is denoted by $\|z\|_2 = \sqrt{\ip{z}{z}}$, and for any convex set $\Cscr$, $$\proj_{\Cscr}(z) \coloneqq \argmin{x \in \Cscr}\ \|x - z\|_2$$ gives the unique orthogonal projection of $z$ onto $\Cscr$. Let $\cone(\Cscr) \coloneqq \set{\alpha x | \alpha \geq 0,\ x \in \Cscr}$ denote the convex cone generated by $\Cscr$.  For any atomic set $\Ascr \subseteq \Re^n$, let $\Dscr(\Ascr, z) \coloneqq \cone\set{d\in \Re^n \mid \gauge\As(z + d) \leq \gauge\As(z)}$ denote the descent cone of $\Ascr$ at $z$. The face of $\Ascr$ exposed by $z$ is the set
\begin{equation*}
  \Fscr(\Ascr, z) = \conv\Set{x\in\Ascr | \ip x z = \sup_{u\in\Ascr}\ip u z},
\end{equation*} 
which is the convex hull of all elements in $\Ascr$ that lie on the supporting hyperplane defined by the normal $z$.

Let $\Nscr(0,I)$ denote the standard Gaussian distribution. For a convex cone $\Dscr$, let $\delta(\Dscr) \coloneqq \mE_{g\sim\Nscr(0,I)} \|\proj_{\Dscr}(g)\|_2^2 $ denote the statistical dimension of $\Dscr$. For any compact set $\Cscr$, let $\Uscr(\Cscr)$ denote the uniform distribution over $\Cscr$. 

The following blanket assumption holds throughout the paper.
\begin{assumption}[Measurement model]\label{assume-blanket}
  The linear model~\eqref{eq:main-problem} satisfies the following conditions: the linear map
  $M:\Re^n\to\Re^m$ has i.i.d.\@ standard Gaussian entries; the noise vector $\eta$ satisfies $\|\eta\|_2\leq \alpha$ for some scalar $\alpha$; and the relative signal powers $\{\lambda_i\}_{i=1}^k$ satisfy~\eqref{eq-lambda-def}.
\end{assumption} 

}


\section{Decompressing the superposition}\label{sec:2} 

\begin{figure}[t]
  \centering
 \includegraphics[width=.8\linewidth, page=7]{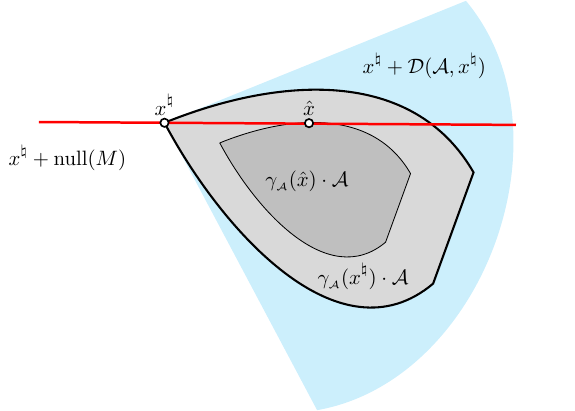} 
\caption{A non-trivial intersection of $\Dscr(\Ascr, x\na)$ and $\Null(M)$ is required for successful decompression. The blue shaded region represents the shifted descent cone $x\na + \Dscr(\Ascr, x\na)$, and red line represents the shifted null space $\Null(M) + x\na$. If $\Dscr(\Ascr, x\na) \cap \Null(M) \neq \{0\}$ (as depicted here) then there exists a vector $\hat x$ such that  $\gauge\As(\hat x) < \gauge\As(x\na)$ and $M\hat x = Mx\na$.\label{fig:descent-cone}}
\end{figure}

As shown in \Cref{sec:decompression-deconvolution}, under the assumption that the individual signals $x\nai$ are $\Ascr_i$ sparse, the superposition $\xagg$ is sparse with respect to the aggregate atomic set $\Ascr\subs$.
Thus, the decompression of the observations $b$ in \eqref{eq:main-problem} is accomplished by minimizing the gauge to $\Ascr\subs$ to within the bound on the noise level $\|\eta\|_2\le\alpha$, as modeled by the recovery problem \eqref{eq:decompression}.
Without noise (i.e, $\alpha=0$), the aggregate signal $\xagg$ is the unique solution to \eqref{eq:decompression} when the null space of the measurement operator $M$ has only a trivial intersection with the descent cone $\Dscr\subs\coloneqq\Dscr(\Ascr\subs, \xagg)$. 
In other words, $\xagg$ is the unique solution of \eqref{eq:decompression} if and only if 
\begin{equation}\label{eq-unique-optimality}
  \Dscr\subs\cap\Null(M) = \{0\}.  
\end{equation}
\cref{fig:descent-cone} illustrates the geometry of this optimality condition, and depicts a case in which it doesn't hold.

If the linear operator $M$ is derived from Gaussian measurements, Gordon~\cite{gordon1988milman} characterized the probability of the event~\eqref{eq-unique-optimality} as a function of the Gaussian width of the descent cone $\Dscr\subs$ and the number of measurements $m$. This result is the basis for recovery guarantees developed by Chandrasekaran et al.~\cite{chandrasekaran2012convex} and Tropp~\cite{tropp2015convex} for a convex formulation similar to \eqref{eq:decompression}.

Intuitively, the number of measurements required for stable recovery of the superposition $\xagg$ depends on the total complexity of the $k$ constituent $\Ascr_i$-sparse vectors $x\nai$. The complexity is measured in terms of the statistical dimension of each of the descent cones $\Dscr_i$.
Tropp~\cite[Corollary~3.5]{tropp2015convex} established a bound on the recovery error between the solutions of the decompression problem~\eqref{eq:decompression} and the superposition $x\nag$ that depends on the statistical dimension $\delta(\Dscr\subs)$ of its descent cone. The following proposition is a restatement of Tropp~\cite[Corollary~3.5]{tropp2015convex} applied to the decompression problem \eqref{eq:decompression}.
\begin{proposition}[Stable decompression of the aggregate]%
  \label{thm:tropp}
  For any $t>0$, any solution $x^*$ of \eqref{eq:decompression} satisfies
  \begin{equation*}
      \|x^* - \xagg\|_2 \leq 2\alpha\left[\sqrt{m - 1} - \sqrt{\delta(\Dscr\subs)} - t\right]_+^{-1}
  \end{equation*}
  with probability at least $1 - \exp(-t^2/2)$, where $[\xi]_+=\max\{0,\xi\}$. 
\end{proposition}
The statistical dimension of $\Dscr\subs$ is in general challenging to compute. However, we show in \cref{sec:3.1} that when all the signals $\{x\nai\}_{i=1}^k$ are incoherent, a reasonable upper bound on $\delta(\Dscr\subs)$ can be guaranteed; see \cref{coro:bound_sta_dim}.

\change{As we can see from~\cref{thm:tropp}, the recovery error bound depends linearly on the noise level $\alpha$. This result relies on the assumption that the noise level is overestimated, i.e., $\alpha \geq \|\eta\|_2$, which is part of \cref{assume-blanket}. However, when the noise level is underestimated, i.e., $\alpha < \|\eta\|_2$, we can not provide any meaningful recovery error bound because even the ground-truth signal may not be feasible for the decompression problem \eqref{eq:decompression}. This limitation suggests that if the true noise level isn't known in practice, then we can start with a relative large $\alpha$, and keep reducing it until satisfactory results are obtained.}


\section{Deconvolving the components}\label{sec:3}

The second stage of our approach is the deconvolution stage which separates the recovered aggregate signal into its constituent components. In order to successfully separate the superposition $x\nag$ into its components $\{x\nai\}_{i=1}^k$ using the deconvolution problem \eqref{eq:deconvolution}, additional assumption on dissimilarity between the atomic representations of the individual signals is generally required. For example, it can be challenging to separate the superposition of two sparse signals or two low-rank signals without additional assumptions. We follow McCoy and Tropp~\cite{mccoy2013achievable}, and measure the dissimilarity between signal structures---and thus their incoherence---using the angles between corresponding descent cones.


\begin{figure*}[h]
  \centering
  \begin{tabular}{@{}cccc@{}}
    \includegraphics[width=.20\textwidth, page=3]{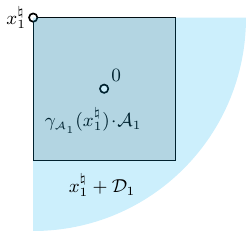}
  & \includegraphics[width=.20\textwidth, page=4]{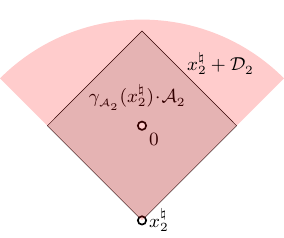}
  & \includegraphics[width=.20\textwidth, page=5]{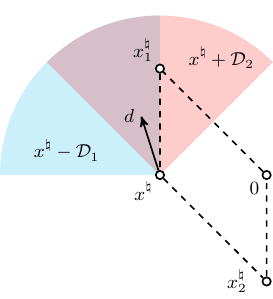}
  & \includegraphics[width=.20\textwidth, page=6]{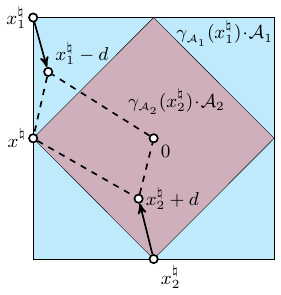}
\\[6pt] \small (a) $\Dscr_1\coloneqq\Dscr(\Ascr_1, x_1\na)$ & (b) $\Dscr_2\coloneqq\Dscr(\Ascr_2, x_2\na)$ & (c) $d \in -\Dscr_1\cap\Dscr_2$ & (d) $x_1\na - d$ and $x_2\na + d$ 
  \end{tabular}
  \caption{The top row depicts two scaled atomic sets $\gauge\Aso(x_i\na)\cdot\Ascr_i$ and the corresponding descent cones $x_i\na + \Dscr_i$ (shifted to lie at $x_i\na$) for $i=1,2$. (c) The descent cones shifted to $x\na = x_1\na + x_2\na$, with $\Dscr_1$ negated; the vector $d$ lies in their intersection. (d) The vector $d$ descends on both scaled atomic sets, so that $\gauge\Aso(x_1\na - d) < \gauge\Aso(x_1)$ and $\gauge\Ast(x_2\na + d) < \gauge\Ast(x_2\na)$.\label{fig:angle_cone}}
\end{figure*}

To motivate the incoherence definition, consider the case where there are only $k=2$ signals $x\nao$ and $x\nat$. If the descent cones $-\Dscr_1$ and $\Dscr_2$ have a nontrivial intersection, then there exists a nonzero direction $d \in -\Dscr_1 \cap \Dscr_2$ such that $\gauge\Aso(x\nao - d)<\gauge\Aso(x\nao)$ and $\gauge\Ast(x\nat + d)<\gauge\Ast(x\nat)$, which contradicts the optimality condition required for $x\nao$ and $x\nat$ to be unique minimizers of \eqref{eq:deconvolution}.  Thus, deconvolution only succeeds if the descent cones have a trivial intersection, which can be characterized using angle between the descent cones. \Cref{fig:angle_cone} illustrates this geometry. 

Obert~\cite{obert1991angle} defined the angle between two cones $\Kscr_1$ and $\Kscr_2$ in $\Re^n$ as the minimal angle between vectors in these two cones.
It follows that the cosine of the angle between two cones can be expressed as 
\begin{align*}
  \cos\angle(\Kscr_1, \Kscr_2) 
  = \sup\set{\ip{x}{y} | x \in \bar\Kscr_{1},\ y \in\bar\Kscr_{2}},
\end{align*}
where
\begin{equation*}
  \bar\Kscr_{i} := \Kscr_i\cap\mS^{n - 1},\ i=1,2.
\end{equation*}
For the general case where the number of signals $k\ge2$, a natural choice for a measure of incoherence between these structured signals is the minimum angle between the descent cone of a signal with respect to the remaining descent cones.  

\begin{definition} \label{def:incoherence}
  The pairs $\{(x\nai, \Ascr_i)\}_{i=1}^k$ are $\beta$-incoherent with $\beta\in(0,1]$ if for all $i\in\irange1k$,
  \[
    \cos\angle\left({-\Dscr_i},\, \sum_{j \neq i}\Dscr_j\right)
    \leq 1 - \beta. 
  \]
\end{definition}

We use the incoherence between descent cones to bound the error between the true constituent signals $\{x\nai\}_{i=1}^k$ and the solution set of the deconvolution problem~\eqref{eq:deconvolution}. This bound depends on the accuracy of the approximation $x\subs^*$ to the true superposition $x\nag$ and is shown in \cref{thm-stability}.
\begin{proposition}[Stable deconvolution]\label{thm-stability} 
  If the pairs $\{(x\nai, \Ascr_i)\}_{i=1}^k$ are $\beta$-incoherent for some $\beta\in(0,1]$, then any set of solutions $\{x_i^*\}_{i=1}^k$ of~\eqref{eq:deconvolution} satisfies for all $i\in\irange1k$
  \begin{align*}
    \|x_i^* - x\nai\|_2 \leq \|x\subs^*-\xagg\|_2/\sqrt{\beta},
  \end{align*}
  where $x\subs^*$ is any solution of~\eqref{eq:decompression}.
\end{proposition}
In summary, a large angle between negation of a descent cone $-\Dscr_i$ and all the other descent cones---as reflected by a large incoherence constant $\beta$---corresponds a small error between each $x_i^*$ and the ground truth $x\nai$.

\subsection{Bound on \texorpdfstring{$\delta(\Dscr\subs)$}{Ds} under incoherence}\label{sec:3.1}
\Cref{thm:tropp} gives a stable recovery result for the decompression stage. However, the recovery bound depends on the the statistical dimension of $\Dscr\subs$, which is challenging to compute even when the statistical dimension of the individual descent cone $\Dscr_i$ is known. In this section, we show that the incoherence between the structured signals $\{x\nai\}_{i=1}^k$ is sufficient to establish an upper bound for $\delta(\Dscr\subs)$. 
We start with the $k = 2$ case. \cref{prop:bound_sta_dim} shows that if the angle between two cones is bounded, then the statistical dimension of the sum of these two cones is also bounded. 
\begin{proposition}[Bound on statistical dimension of sum]\label{prop:bound_sta_dim}
  Let $\Kscr_1$ and $\Kscr_2$ be two closed convex cones in $\Re^n$. If $\cos\angle(-\Kscr_1, \Kscr_2) \leq 1 - \beta$ for some $\beta \in (0, 1]$, then 
  \[\sqrt{\delta(\Kscr_1 + \Kscr_2)} \leq \tfrac{1}{\sqrt{\beta}}\left(\sqrt{\delta(\Kscr_1)} + \sqrt{\delta(\Kscr_2)} \right).\]
\end{proposition}

This result generalizes to an arbitrary number of cones. 
\begin{corollary}[Bound on statistical dimension of sum under incoherence]\label{coro:bound_sta_dim}
   If the pairs $\{(x\nai, \Ascr_i)\}_{i=1}^k$ are $\beta$-incoherent for some $\beta\in(0,1]$, then
   \[\sqrt{\delta(\Dscr\subs)} \leq \beta^{-\tfrac{k-1}{2}}\sum_{i=1}^k\sqrt{\delta(\Dscr_i)}.\]
\end{corollary} 
\cref{coro:bound_sta_dim} shows that when the pairs $\{(x\nai, \Ascr_i)\}_{i=1}^k$ are $\beta$-incoherent, $\delta(\Dscr\subs)$ can be upper bounded in terms of the statistical dimension of individual descent cones.


\section{Inducing incoherence through random rotation}\label{sec-incoherence}

\Cref{thm-stability} establishes the stability of the deconvolution problem in the case that the unknown signals are $\beta$-incoherent, as formalized in \cref{def:incoherence}. However, except in very special cases like randomly rotated signals, it is not feasible to determine the incoherence constant $\beta$. We build on McCoy and Tropp's random rotation model~\cite{mccoy2013achievable} to quantify, with high probability, the $\beta$-incoherence of $k$ randomly-rotated atomic sparse signals, and present a recovery result for a randomly rotated case.

We first consider a simpler case of two general cones, one of which is randomly rotated. Let $\SO(n)$ denote the special orthogonal group, which consists of all $n$-by-$n$ orthogonal matrices with unit determinant. The following proposition provides a probabilistic bound on the angle between the two cones in terms of their statistical dimension. This geometric result maybe of intrinsic interest in other contexts. 
\begin{proposition}[Probabilistic bound under random rotation] \label{prop-angle-cones}
   Let $Q$ is drawn uniformly at random from $\SO(n)$. Let $\Kscr_1$ and $\Kscr_2$ be two closed convex cones in $\Re^n$. For any $t \geq 0$, we have
\begin{align*}
  \mP\bigg[\cos\angle(\Kscr_1, Q\Kscr_2) &\geq  \tfrac{3}{\sqrt{n}}\left(\sqrt{\delta(\Kscr_1)} + \sqrt{\delta(\Kscr_2)}\right) + t\bigg]
  \\&\leq \exp(-\tfrac{n-2}{8}t^2).
\end{align*}
\end{proposition}

We now assume that the $k$ structured signals $x\nai$ are defined via a random rotations of $k$ underlying structured signals $\hat x_i$.
\begin{assumption}[Random rotations]\label{assume-random-rotation}
   Fix $\hat x_i$ and $\hat \Ascr_i$ for $i\in\irange1k$ such that $\hat x_i$ is sparse with respect to atomic set $\hat\Ascr_i$. For each $i\in\irange1k$, assume 
  \begin{equation*}
    x\nai \coloneqq Q_i  \hat x_i \quad\mbox{and}\quad \Ascr_i \coloneqq  Q_i \hat\Ascr_i,
  \end{equation*}
  where the matrices $Q_i$ are drawn uniformly and i.i.d.\@ from $\SO(n)$.
\end{assumption}

Our next proposition 
shows that, under mild conditions, randomly rotated structured signals are incoherent with high probability.
\begin{proposition} \label{thm:Incoherence}
    Suppose that \cref{assume-random-rotation} holds. If $\sum_{i=1}^k \sqrt{\delta(\Dscr_i)} \leq \left(1 - 4^{- \scaleto{\tfrac{1}{k-1}}{10pt} } - t\right)\sqrt{n} / 6$ for some $t>0$, then the rotated pairs $\{(x\nai,\Ascr_i)\}_{i=1}^k$ are $4^{- \scaleto{\tfrac{1}{k-1}}{10pt} }$-incoherent with probability at least $1 - k(k-1)\exp(-\tfrac{n-2}{8}t^2)$.
\end{proposition}
\cref{thm:Incoherence} requires $\sum_{i=1}^k \sqrt{\delta(\Dscr_i)}$ to scale as $\sqrt{n}$ and thus controls the total complexity of the $k$ unknown signals. We now state the main theorem and show that randomly rotated vectors can be recovered using the two-stage approach~\eqref{eq:decompression} and~\eqref{eq:deconvolution}.
\begin{theorem} \label{corollary-stability}
   Suppose that \cref{assume-blanket,assume-random-rotation} hold. For any $t_1, t_2 > 0$, if $\sum_{i=1}^k \sqrt{\delta(\Dscr_i)} \leq \left(1 - 4^{- \scaleto{\tfrac{1}{k-1}}{10pt} } - t_2\right)\sqrt{n} / 6$, then any set of minimizers $\{x_i^*\}_{i=1}^k$ of~\eqref{eq:deconvolution} satisfies 
 \begin{equation} \label{eq:theory_bound}
   \|x_i^* - x\nai\|_2 
    \leq
    \frac{4\alpha}{\left[\sqrt{m-1} - c\sum_{i=1}^k \sqrt{\delta(\Dscr_i)} - t_1\right]_+}
 \end{equation}
for all $i\in\irange1k$ with probability at least $$1 - \exp\left(-t_1^2/2\right) - k(k-1)\exp(-\tfrac{n-2}{8}t_2^2)$$  with $c\leq2$.
\end{theorem}
The proof follows directly from \cref{thm:tropp}, \cref{thm-stability}, \cref{coro:bound_sta_dim}, \cref{thm:Incoherence}, and the probability union bound. We verify empirically in \cref{sec:6.1} the tightness of the bound in~\eqref{eq:theory_bound}.


\subsection{Comparison of error bound} \label{sec:comparasion}
Here we compare our results to the one provided in \cite{mccoy2013achievable}, which also developed a novel procedure to solve the demixing problem \eqref{eq:main-problem}. McCoy and Tropp~\cite{mccoy2013achievable} introduced the constrained optimization problem
\begin{equation}
  \label{eq:achievable-model}
  \begin{array}{ll}
  \minimize{x_1,\ldots,x_k}
  & \left\|
    M^\dagger\left(M\sum_{i=1}^k x_i - b\right)
  \right\|_2 \\
 \st
  & \gauge(x_i) \le \gauge\Asi(x\nai), \ \forall i\in\irange1k,
  \end{array}
\end{equation}
where $M^\dagger$ is the Moore-Penrose pseudo-inverse of $M$. They showed that if \(n\ge m\ge \sum_{i=1}^k\delta(\Dscr_i)+\BigOh(\sqrt{kn})\) and $\{x\nai\}_{i=1}^k$ are randomly rotated as per \cref{assume-random-rotation}, then any set of minimizers $\{x_i^*\}_{i=1}^k$ of~\eqref{eq:achievable-model} satisfies with high probability the bound
\begin{equation}\label{eq:Tropp_error}
  \|x_i^*-x\nai\|_2\le C\|M^\dagger\eta\|_2
\end{equation}
for all $i\in\irange1k$~\cite[Theorem~A]{mccoy2013achievable}. To our knowledge, this result is the first to show that stable recovery of the constituent signals $\{x\nai\}_{i=1}^k$ is possible with high probability provided the number of measurement grow linearly in $k$. However, the constant $C$ in the error bound \eqref{eq:Tropp_error} could depend on all of the problem parameters except $\eta$. As a comparison to \cref{corollary-stability}, the error bound in \eqref{eq:theory_bound} makes explicit the effect of all problem parameters.


\section{Decompression and deconvolution algorithm}\label{sec-algorithms}
\def\t{^{(t)}}
We describe a procedure for obtaining solutions for the decompression~\eqref{eq:decompression} and deconvolution~\eqref{eq:deconvolution} problems. The procedure first solves the decompression problem~\eqref{eq:decompression} using an algorithm that doesn't store or track an approximation to $x\nag$, which in many contexts may be too large to store or manipulate directly. Instead, the algorithm produces a sequence of iterates $r\t\coloneqq b-Mx\t$ that approximate the residual vector corresponding to an implicit approximation $x\t$ of $x\nag$. The procedure requires only the storage of several vectors of length $m$, which represents the size of the data $b$. As we show in \cref{sec-deconv-algo}, the solution to the deconvolution problem~\eqref{eq:deconvolution} is subsequently obtained via an unconstrained linear least-squares problem that uses information implicit in this residual vector. \cref{alg-level-set} summarizes the overall procedure. 

\newcommand{\vupp}{v_{\sf hi}}
\newcommand{\vlow}{\ell}

\begin{algorithm}[t]
 \DontPrintSemicolon
 \SetKwComment{tcp}{\tiny [}{]}
 \caption{Decompression and deconvolution algorithm\label{alg-level-set}}
 \smallskip
 \KwIn{noise level $\alpha>0$; accuracy $\epsilon>0$}
 $\tau^{0} \gets 0$

 \For(\tcp*[f]{\tiny level-set iterations}){\nllabel{alg-level-set-loop}$t\gets0,1,2,\ldots$}{
   $(r\t,\, p\t,\, \vlow\t)\leftarrow\texttt{DCG}(\tau\t)$ \tcp*{\tiny solve~\eqref{eq:value-fn} approximately}\nllabel{alg-level-set-cg}
   \lIf(\tcp*[f]{\tiny test $\epsilon$-infeasibility}){$\|r\t\| > \sqrt{\alpha^2 + \epsilon}$}{break}
   $\tau\tp1 \gets \tau\t + \frac{\vlow\t-\alpha^2/2}{\ip{p\t}{r\t}}$\nllabel{alg-level-set-update}
   \tcp*{\tiny Newton update}
 }
  $(x_1,\ldots,x_k) \gets \mbox{solve~\eqref{eq-primal-recovery}}$ \nllabel{alg-primal-recovery}
  \tcp*{\tiny solve \eqref{eq:deconvolution}}
 \Return{$(x_1,\ldots,x_k)$}
\end{algorithm}

\subsection{Level-set method} \label{sec:level}

The loop beginning at \cref{alg-level-set-loop}  of \cref{alg-level-set} describes the level-set procedure for solving the decompression problem \eqref{eq:decompression}~\cite{BergFriedlander:2008,berg2011sparse,aravkin2016levelset}. More specifically, it approximately solves a sequence of problems
\begin{equation} \label{eq:value-fn}
  v(\tau) \coloneqq \min_{x}\left\{\,\half\|Mx - b\|^2 \mid \gauge\Ass(x) \leq \tau\,\right\},
\end{equation}
parameterized by the scalar $\tau$ that defines the level-set constraint. The subproblem \eqref{eq:value-fn} is solved by the dual conditional gradient method (\cref{alg-level-set-cg} of \cref{alg-level-set}), introduced in \cref{sec:dcg}. 
Under modest assumptions satisfied by this problem, the sequence $\tau\t\to\tau_*=\opt$, the optimal value of~\eqref{eq:decompression}. The tail of the resulting sequence of computed solutions to~\eqref{eq:value-fn} is super-optimal and $\epsilon$-infeasible for~\eqref{eq:decompression}, i.e., a solution $x$ satisfies 
\begin{equation} \label{eq:super-optimal}
  \gauge\Ass(x) \leq \opt \enspace\text{and}\enspace \|Mx - b\| \leq \sqrt{\alpha^2 + \epsilon},
\end{equation}
where $\epsilon$ is a specified optimality tolerance. The level-set algorithm requires $\BigOh(\log(1/\epsilon))$ approximate evaluations of the optimization problem~\eqref{eq:value-fn} to achieve this optimality condition. Each approximate evaluation provides a global lower-minorant of $v$ that is used by a Newton-like update to the level-set parameter $\tau\t$; see line~\ref{alg-level-set-update}.

\subsection{Dual conditional gradient method} \label{sec:dcg}

\IncMargin{.3em}
\begin{algorithm}[t]
  \DontPrintSemicolon
  \setcounter{AlgoLine}{0}
  \SetKwComment{tcp}{\tiny [}{\tiny ]}
  \caption{Dual conditional gradient method: \texttt{DCG}($\tau$). This algorithm solves~\eqref{eq:decompression} without reference to the primal iterate $x\t$, and instead returns the implied residual $r\t\equiv b-Mx\t$.\label{alg-vanilla-fw}}
  \KwIn{$\tau$}

  $r^{(0)} \gets b$;\ $q^{(0)} \gets 0$\;

  \For{$t\gets0,1,2,\ldots$}{
     $p\t\in\tau \Fscr(M\Ascr_s;\,r\t)$ \tcp*{\tiny see~\eqref{eq:face-identity}}\nllabel{algo-expose-atom}
     $\Delta r\t\gets p\t-q\t$ \tcp*{\tiny $\Delta r\t\equiv M(a\t-x\t)$}
     $\rho\t \gets \ip{r\t}{\Delta r\t}$\tcp*{\tiny optimality gap}
     \lIf(\tcp*[f]{\tiny break if optimal}){$\rho\t<\epsilon$}{
       break
     }
     $\theta\t \gets \min\set{1,\, \rho\t / \|\Delta r\t\|_2^2}$\tcp*{\tiny exact linesearch}
     $r\tp1 \gets r\t-\theta\t\Delta r\t$\tcp*{\tiny $r\tp1\equiv b-Mx\tp1$}
     $q\tp1 \gets q\t+\theta\t\Delta r\t$\tcp*{\tiny $q\tp1\equiv \phantom{b-{}}Mx\tp1$}
  }
  $\ell\t\gets\half\norm{r\t}^2-\rho\t$\tcp*{\tiny lower bound on optimal value}

  \Return{$r\t$, $p\t$, $\vlow\t$}
\end{algorithm}

The level-set subproblems are solved approximately using the dual conditional-gradient method described by \cref{alg-vanilla-fw}. An implementation of this algorithm requires storage for three $m$-vectors
\[
  p\t\coloneqq Ma\t, \quad q\t\coloneqq Mx\t, \quad r\t\coloneqq b-Mx\t,
\]
(The fourth vector $\Delta r\t$ can be easily operated on implicitly.) Implicit in these vectors are the iterate $x\t$ and current atom $a\t\in\Ascr_s$, which in some situations are prohibitively large to store or manipulate The main computational cost is in Line~\ref{algo-expose-atom}, which uses the residual $r\t$ to expose an atom in the face
$\Fscr(M\Ascr_s;\,r)$ of the mapped atomic set $M\Ascr\subs\subset\Real^m$. Because the exposed faces decompose under set addition, it follows from the expression~\eqref{eq:weighted-atomic-sum} of $\Ascr_s$ that $\Fscr(M\Ascr_s;\, r) = \sum_{i=1}^k\Fscr(\lambda_i M\Ascr_i;\,r)$.  Thus, the facial exposure operation on Line~\ref{algo-expose-atom} can be computed by separately exposing faces on each of the individual mapped atomic sets, which can be implemented in parallel, i.e.,
\begin{equation*}
    p\t = \tau\sum_{i = 1}^k \lambda_i p_i\t
    \text{where}
    p_i\t \in \Fscr(M\Ascr_i;\,r\t) \enspace\forall i\in\irange1k.
\end{equation*}

\change{%
Note that \cref{algo-expose-atom} of \cref{alg-vanilla-fw} can alternatively be implemented using the identity
\begin{equation}\label{eq:face-identity}
 \Fscr(M\Ascr_s;\,r) = M\Fscr(\Ascr_s;\,M^*r);
\end{equation}
see Fan et al.~\cite[Section 3]{fan2019alignment}. This formulation is convenient in cases where the operator $M$ can be applied implicitly to elements of the atomic set $\Ascr_s$.%
}

The conditional-gradient method converges to the required optimality within $\BigOh(1/\epsilon)$ iterations \cite{jaggi2013revisiting}. Combined with the complexity of the level-set method, we thus expect a total worst-case complexity of $\BigOh(\log(1/\epsilon)/\epsilon)$ iterations to satisfy the optimality condition~\eqref{eq:super-optimal}.

\subsection{Exposing the signals}\label{sec-deconv-algo}

Once \cref{alg-level-set} reaches Line~\ref{alg-primal-recovery}, the residual vector $r\t$ contains information about the atoms that are in the support of each of the approximations $x_i^*$ to the signals $x\nai$. It follows from Fan et al.~\cite[Theorem~7.1]{fan2019alignment} that for all $i\in\irange1k$,
\[
  x_i^*\in\cone\Fscr(\Ascr_i;\, M^*r^*), \quad r^* := b-M\sum_{i=1}^k x_i^*.
\]
Thus, a solution of the deconvolution problem \eqref{eq:deconvolution} can be recovered by solving 
\begin{equation} \label{eq-primal-recovery}
\begin{array}{ll}
\minimize{x_1,\ldots,x_k} & \half\|M\textstyle\sum_{i=1}^k x_i - (b-r^{(t)})\|^2 \\
\st & x_i\in\cone\Fscr(\Ascr_i;\,M^* r^{(t)}),
\end{array}
\end{equation}
which can be implemented as a standard linear least-squares problem over the coefficients of the atoms exposed in each of the atomic sets.

\section{Experiments and novel applications} \label{sec:6}

In~\cref{sec:6.1} we empirically verify \cref{corollary-stability} through a set of synthetic experiments on recovering multiple randomly-rotated sparse signals from noiseless and noisy measurements. Note that the random rotation guarantees incoherence among the unknown signals $\{x\nai\}_{i=1}^k$. We also empirically show that random rotation is not required for successful recovery of a class of unknown signals with different underlying structures. In~\cref{sec:6.2} we separate a sparse signal and sparse-in-frequency signal. In~\cref{sec:6.3} we separate the superposition of three signals: a sparse signal, a low-rank matrix, and noise. In~\cref{sec:6.4} we separate a multiscale low-rank synthetic image.

We implement the algorithm described in \cref{sec-algorithms} in the Julia language \cite{BEKS14}. All the experiments are conducted on a Linux server with 8 CPUs and 64Gb memory.

\subsection{Stability of Demixing} \label{sec:6.1}

We provide three experiments that numerically verify the bounds established by \cref{corollary-stability} to solve the demixing problem \eqref{eq:main-problem}. The experiment draws multiple realizations of a random problem specified over a range of parameters $k$ (number of signals), $m$ (number of measurements), $n$ (signal dimension) and $s$ (the sparsity level for each signal). Each signal $x\nai$ in~\eqref{eq:main-problem} is generated according to \cref{assume-random-rotation}, where each vector $x_i^\circ$ is $s$-sparse with respect to the standard basis. By construction, the atomic sets $i\in\irange1k$ are defined to be 
\[
  \Ascr_i = Q_i \set{\pm e_1, \dots, \pm e_n}\text{where} Q_i\sim\Uscr(\SO(n)).
\]
Amelunxen et al.~\cite[Proposition~4.5]{amelunxen2014living} give an upper bound on the statistical dimension of the descent cone for $(x\nai,\Ascr_i)$, and thus for the descent cone at $(x_i^\circ,\Ascr_i^\circ)$, for  $s$-sparse vectors. We use this bound to approximate the statistical dimension $\delta(\Dscr_i)$ of the descent cone $\Dscr_i$ corresponding to the pair $(x\nai,\Ascr_i)$.  We define the maximum absolute error
\begin{equation} \label{eq:rela_diff}
  \mathop{\tt maxerr}\coloneqq \max_{i\in\irange1k}\ \twonorm{x_i^*-x\nai}.
\end{equation}

\subsubsection{Relation between $m$ and $n$} \label{sec:phase_transition1}
We first show a phase portrait for the noiseless case that verifies the relationship between number of measurement $m$ and signal dimension $n$, as stated in \cref{corollary-stability}. The number of signals is fixed at $k=3$ and the sparsity level is fixed at $s=5$. The phase plot is shown in \cref{fig:phase_transition1}, where the horizontal axis represents the signal dimension $n\in\{50, 65, \dots, 500\}$ and the vertical axis represents the number of measurements $m\in\{50, 65, \dots, 500\}$. The colormap indicates the empirical probability of successful demixing over 50 trials, where we say the demixing is successful if $\mathop{\tt maxerr} < 10^{-2}$. The red solid curve and the blue dashed line, respectively, approximate the graphs of the functions
\[\sqrt{m} = \sum_{i=1}^k \sqrt{\delta(\Dscr_i)}\text{and}\sqrt{n} = \sum_{i=1}^k \sqrt{\delta(\Dscr_i)}.\]
The statistical dimensions of $\Dscr_i$ are approximated using ~\cite[Proposition~4.5]{amelunxen2014living}, as stated above. The area above the red curve and to the right of the dashed line corresponds to problem parameters with successful recovery and corroborates the bounds stated in  \cref{corollary-stability}.

\begin{figure}[t]
  \centering\small
  \includegraphics[width=.65\linewidth]{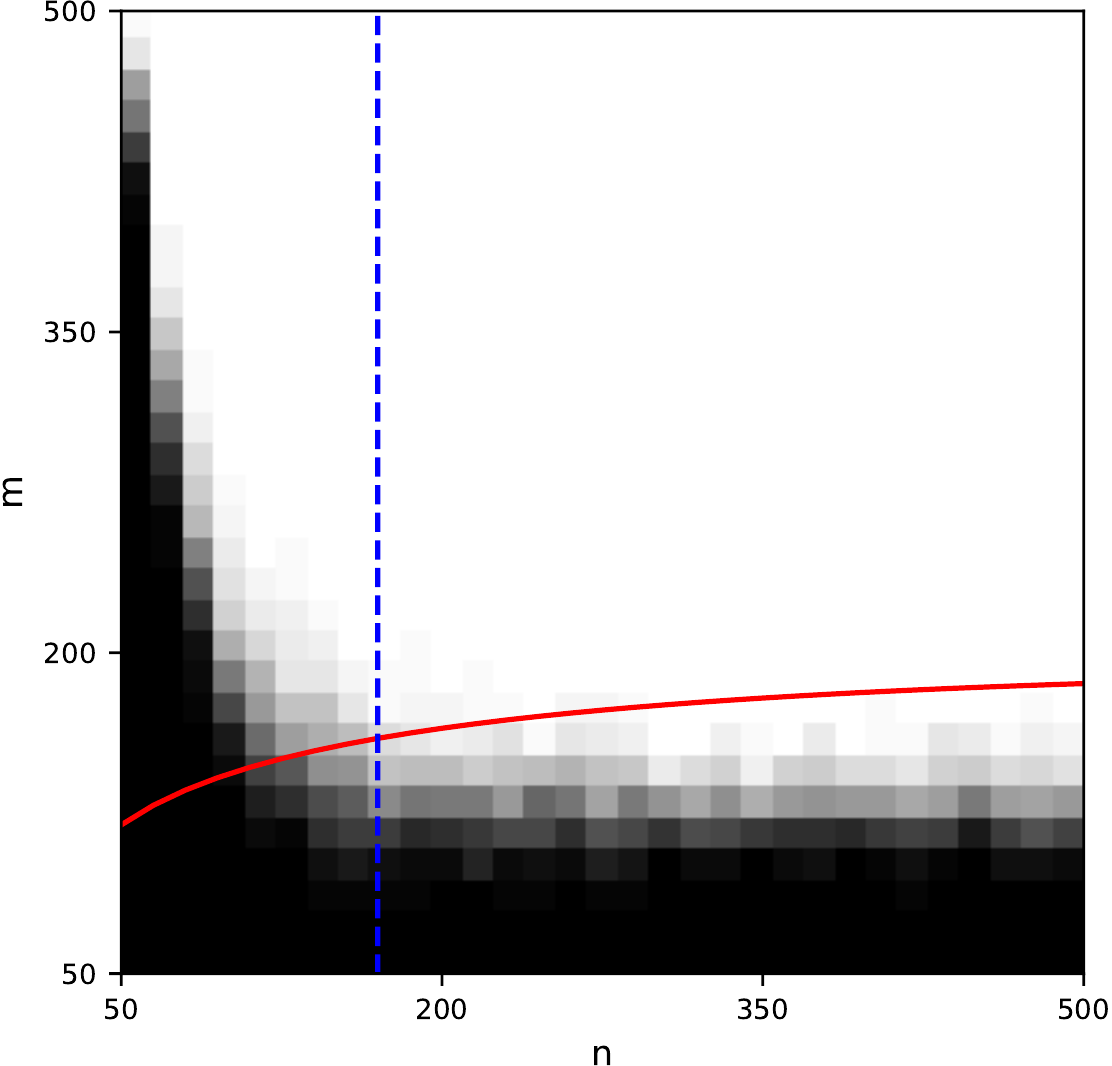}
  \caption{Phase-transition plots for demixing the sum of randomly-rotated sparse signals $\{x\nai\}_{i=1}^k$ from noiseless measurements $b$. The horizontal and vertical axes, respectively, represent the signal dimension $n$ and measurement dimension $m$. The colormap indicates the empirical probability of successful demixing over 50 trials. The red solid curve approximately represents the mapping $\sqrt{m} = \sum_{i=1}^k \sqrt{\delta(\Dscr_i)}$ and the blue dashed line approximately represents the position $\sqrt{n} = \sum_{i=1}^k \sqrt{\delta(\Dscr_i)}$.}
  \label{fig:phase_transition1}
\end{figure}

\subsubsection{Relation between $m$ and $k$}
We also show a phase portrait for the noiseless case that verifies the relationship between number of measurement $m$ and number of signals $k$ stated in \cref{corollary-stability}. The signal dimension is fixed at $n=1000$ and the sparsity level is fixed at $s=3$. The phase plot is shown in \cref{fig:phase_transition2}, where the horizontal axis represents the number of signals $k\in\{2, 3, \dots, 10\}$ and the vertical axis represents the number of measurements $m\in\{100, 200, \dots, 1000\}$. All the other settings are the same as stated in \cref{sec:phase_transition1}. The red line corresponds to $\sqrt{m} = \sum_{i=1}^k \sqrt{\delta(\Dscr_i)}$ and shows that recovery is possible provided the number of measurements scale as $k^2$, when the complexity of all of unknown signals are the same.  

\begin{figure}[t]
  \centering\small
  \includegraphics[width=.65\linewidth]{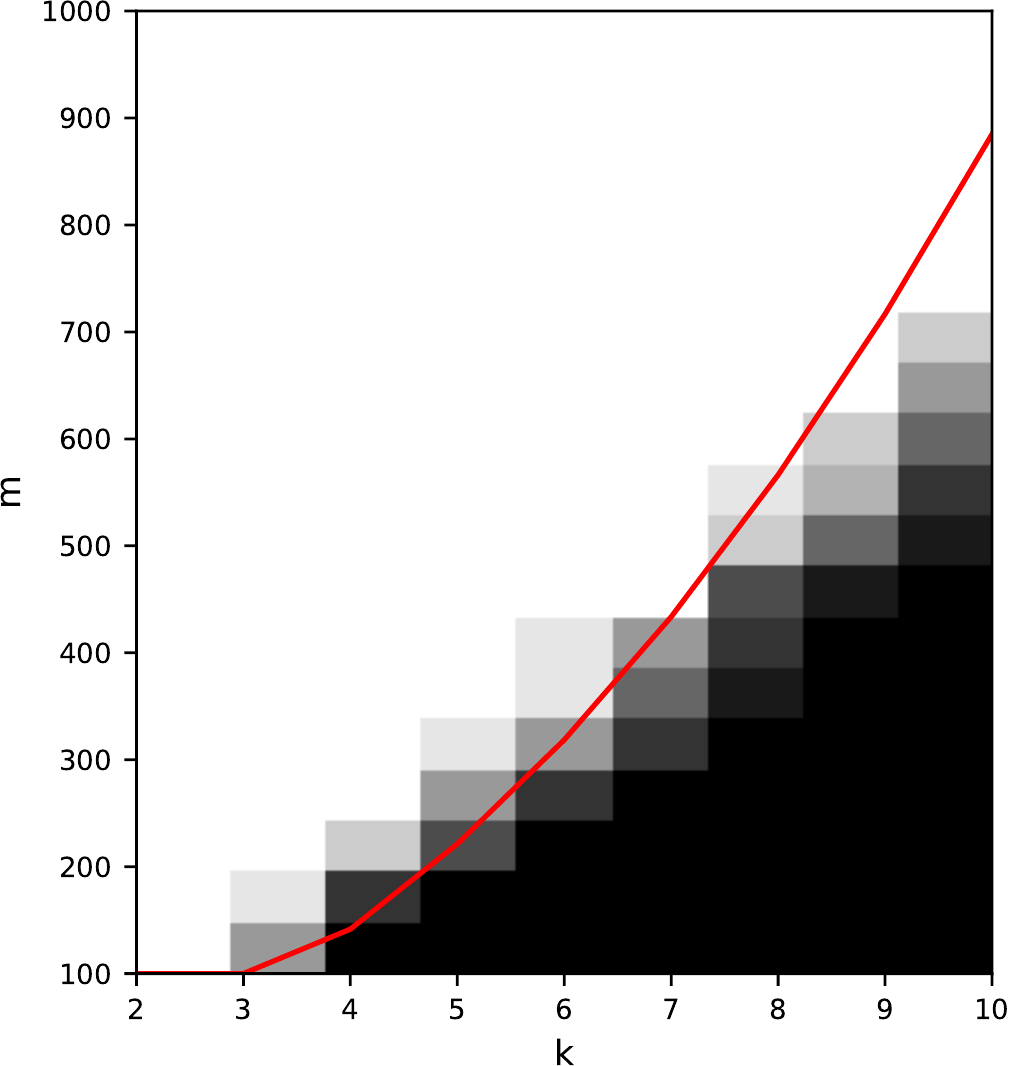}
  \caption{Phase-transition plots for demixing the sum of randomly-rotated sparse signals $\{x\nai\}_{i=1}^k$ from noiseless measurements $b$. The horizontal and vertical axes, respectively, represent the number of signals $k$ and measurement dimension $m$. The colormap indicates the empirical probability of successful demixing over 50 trials. The red solid curve approximately represents the mapping $\sqrt{m} = \sum_{i=1}^k \sqrt{\delta(\Dscr_i)}$.\label{fig:phase_transition2}}
\end{figure}

\subsubsection{Relation between maximal absolute error and noise level}
Lastly, we show a plot for the noisy case that verifies the relationship between maximum absolute error $\mathop{\tt maxerr}$ and noise level $\alpha$ stated in \cref{corollary-stability}. The number of measurement is fixed at $m=125$, the signal dimension is fixed at $n=200$, the number of signals is fixed at $k=3$, and the sparsity level is fixed at $s=5$. The result is shown in \cref{fig:phase_transition3}, where the horizontal axis represents the noise level $\alpha\in\{0.01, 0.02, \dots, 2\}$ and the vertical axis represents the maximum absolute error $\mathop{\tt maxerr}$. The blue curve corresponds to the mean of $\mathop{\tt maxerr}$ over 50 trials and the yellow shaded area corresponds to the standard deviation. The figure verifies the linear dependence of the recovery error with the noise level, as stated in \cref{corollary-stability}.

\begin{figure}[t]
  \centering\small
  \includegraphics[width=.65\linewidth]{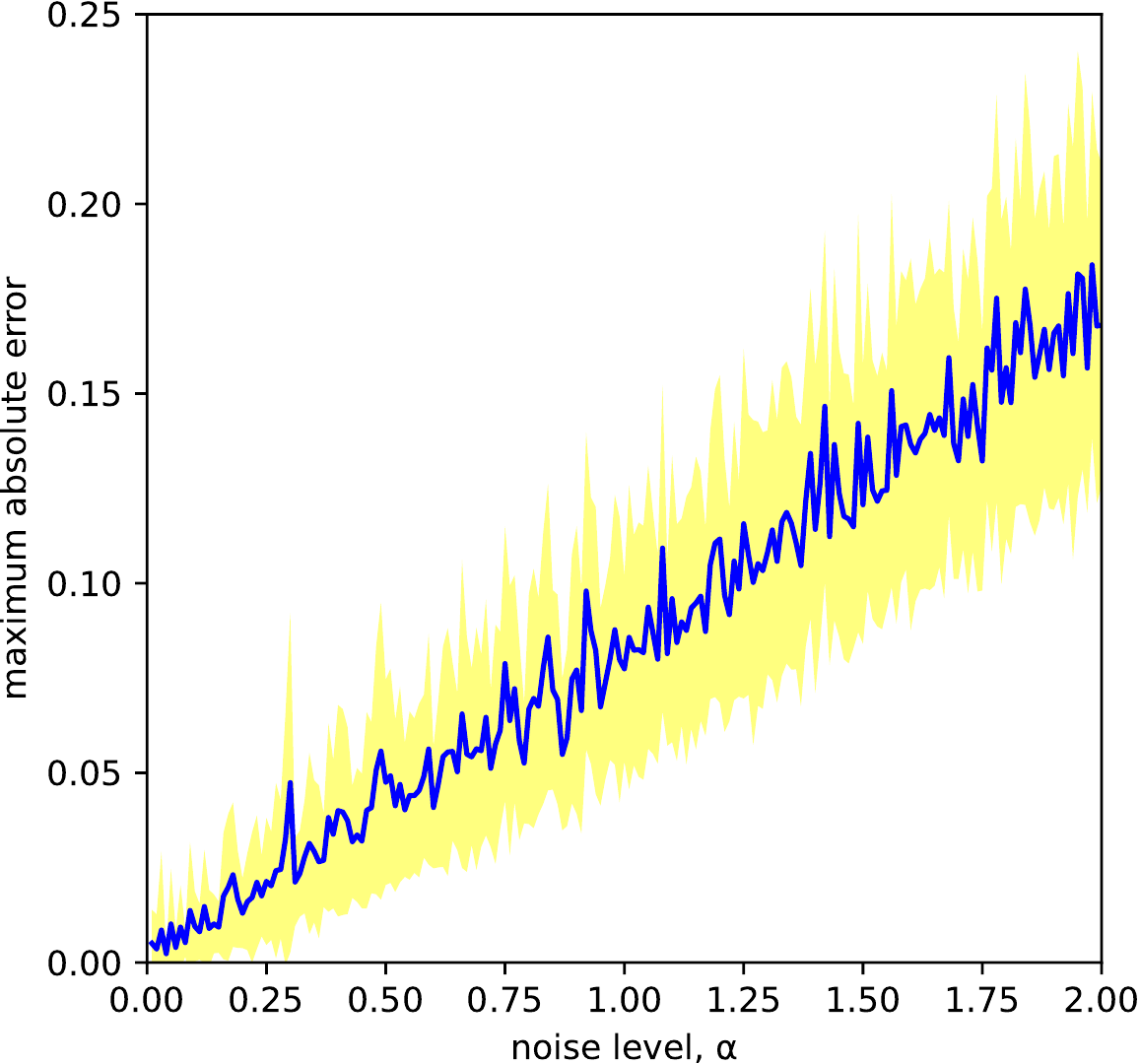}
  \caption{Error-noise plot for demixing the sum of randomly-rotated sparse signals $\{x\nai\}_{i=1}^k$ from noisy measurements $b$. The horizontal and vertical axes, respectively, represent the noise level $\alpha$ and the maximum absolute error $\mathop{\tt maxerr}$. The blue curve indicates the relationship between the empirical average of $\mathop{\tt maxerr}$ over 50 trials and $\alpha$, and the yellow shaded area indicated the empirical standard deviation. }
  \label{fig:phase_transition3}
\end{figure}

\subsection{Separation of sparse and sparse-in-frequency signals} \label{sec:6.2}

\begin{figure}[t]
  \centering
 \includegraphics[width=.9\linewidth]{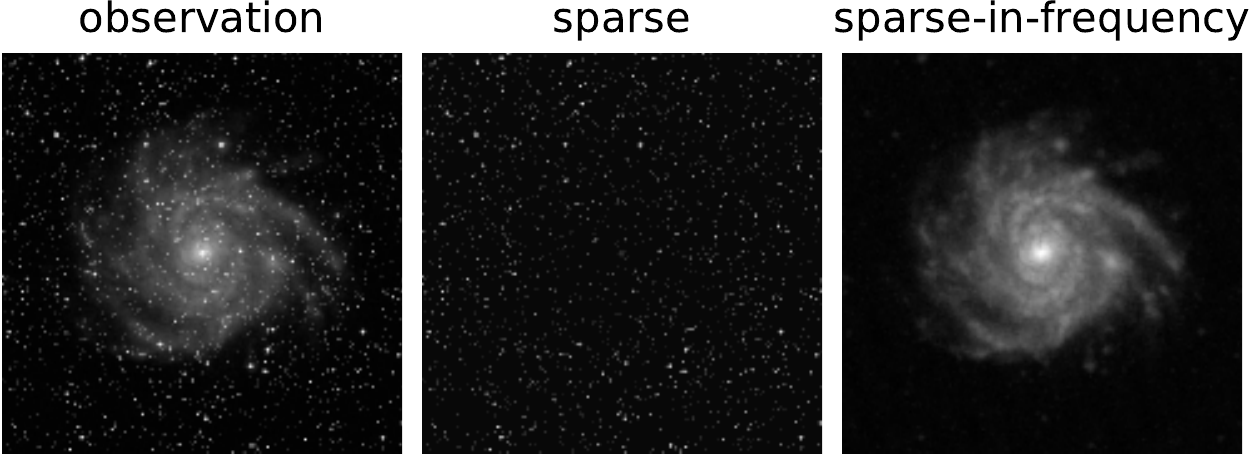}
  \caption{The star-galaxy separation experiment features two distinct signal components. The image size is $601\times601$ pixels.\label{fig:star_galaxy}}
\end{figure}

\change{We reproduce the experiments done by McCoy et al.~\cite{mccoy2014convexity} on separating an astronomical image into sparse and sparse-in-frequency signals.}
An $n$-vector $x$ is sparse-in-frequency if its discrete cosine transform (DCT) $Dx$ is sparse, where the orthogonal linear map $D:\Re^n\to\Re^n$ encodes the DCT. Define the observations and corresponding atomic sets
\[
  b = x\na_s + x\na_d,
  \quad
  \Ascr_s \coloneqq  \set{\pm e_1, \dots, \pm e_n}, \quad \Ascr_d = D^*\Ascr_s.
\]
The star-galaxy image shown in \cref{fig:star_galaxy} exemplifies this superposition: the stars are well-represented by sparse matrices in $\Ascr_s$, and the galaxy component is well-represented by sinusoidal elements in $\Ascr_d$. The image size is $601\times601$. The results of the separation are shown in the second two panels of \cref{fig:star_galaxy}.

\subsection{Sparse and low rank matrix decomposition with structured noise}\label{sec:6.3}

In this next example we decompose an image that contains a sparse foreground, a low-rank background, and structured noise. This is an example of sparse principle component analysis~\cite{fazel1998approximations,fhb01,pati1994phase,valiant1977graph}. Typically, the entry-wise 1-norm and the nuclear norm are used to extract from the matrix each of these qualitatively different structures. Here, we treat the noise as its own signal that also needs to be separated. We consider the observations
\[B = X\na_s + X\na_l + X\na_n,\]
where $X\na_s\in\Re^{m\times n}$ is sparse, $X\na_l\in\Re^{m\times n}$ is low-rank matrix, and $X\na_n\in\Re^{m\times n}$ represents structured noise so that $PX\na_nQ$ is sparse, where $P$ and $Q$ are random orthogonal $m$-by-$m$ matrices. Based on the atomic framework, we choose the atomic sets for $X\na_s$, $X\na_l$, and $X\na_n$, respective, as
\begin{align*}
    \Ascr_s &= \set{\pm E_{i,j} | 1 \leq i \leq m, 1 \leq j \leq n },
  \\\Ascr_l &= \set{uv^\intercal \mid u \in \Re^m,\ v \in \Re^n,\ \|u\|_2=\|v\|_2 = 1},
  \\\Ascr_n &= P^\intercal\Ascr_s Q^\intercal,
\end{align*}
where $E_{i,j}$ is a $m\times n$ matrix with a single nonzero entry $(i,j)$ with value $1$.

\change{%
Although the elements of the atomic sets $\mathcal{A}_s$, $\mathcal{A}_l$ and $\mathcal{A}_n$ are described as explicit matrices, these elements can be exposed and operated on implicitly without ever forming these matrices. Thus, \cref{algo-expose-atom} of \cref{alg-vanilla-fw} can be implemented efficiently for very large examples. In particular, let $Z$ be a fixed exposing matrix. Then the exposed atom from $\mathcal{A}_s$ can be computed by finding the entry in $Z$ with maximum absolute value; the exposed atom from $\mathcal{A}_l$ can be computed by finding the leading singular vectors of $Z$; and exposed atom from $\mathcal{A}_n$ can be computed by finding the entry in $PZQ$ with the maximum absolute value. Fan et al.~\cite{fan2019alignment} provide more detail on how to efficiently implement these operations.%
}

For the numerical experiment, we consider the noisy chess board in-painting problem. The chess foreground is sparse and the chess board background is low rank. The image size is $596\times596$. The experiment result is shown in~\cref{fig:chess_board}. 

\begin{figure}[t]
  \centering
  \includegraphics[width=.6\linewidth]{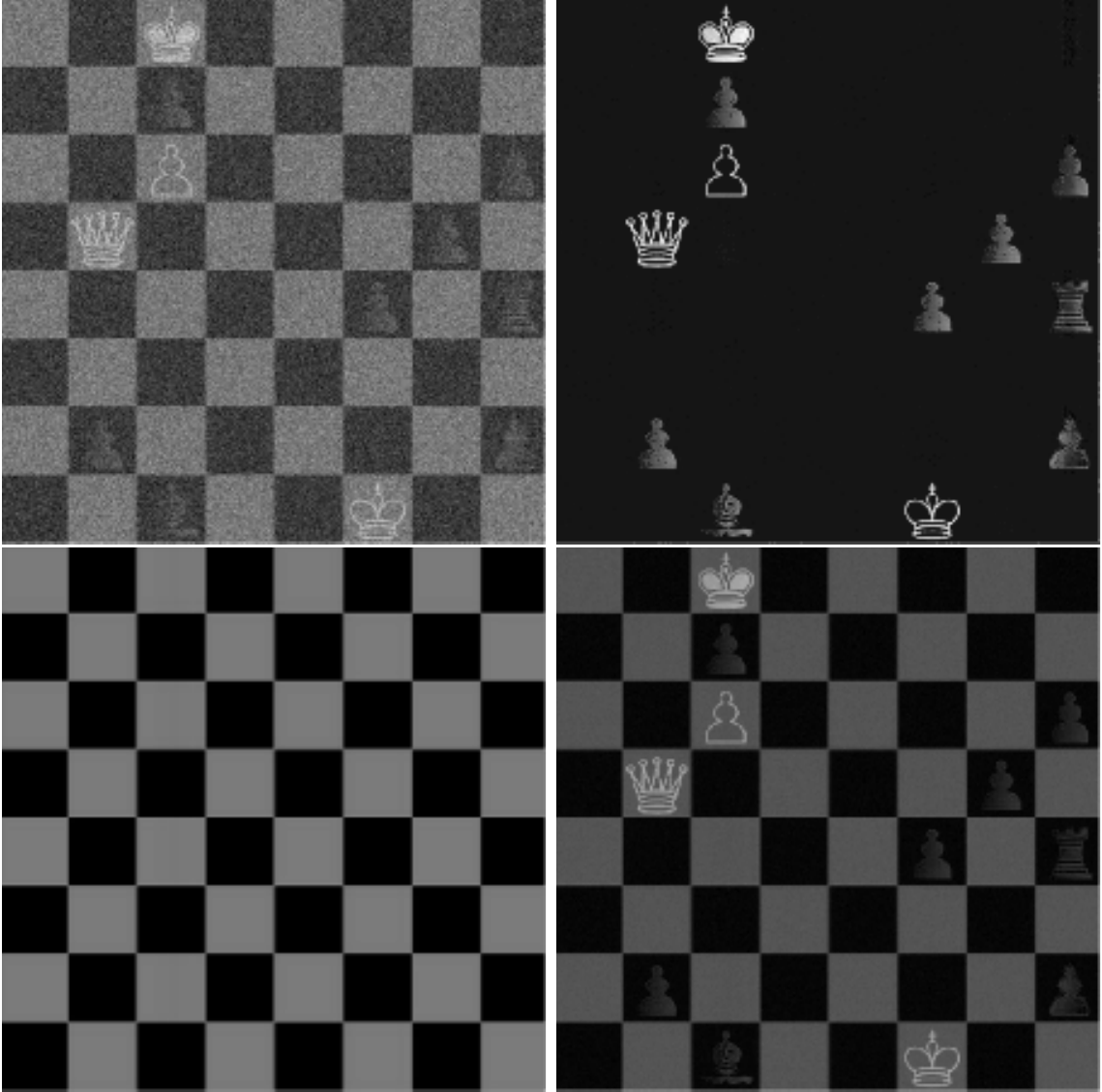}
  \caption{Noisy chess board in-painting experiment. The image size is $596\times596$. Northwest: noisy observations; Northeast: recovered sparse component; Southwest: recovered low rank component; Southeast: denoising result.}
  \label{fig:chess_board}
\end{figure}

\subsection{Multiscale low rank matrix decomposition} \label{sec:6.4}

\begin{figure}[t]
  \centering
  \includegraphics[width=\linewidth]{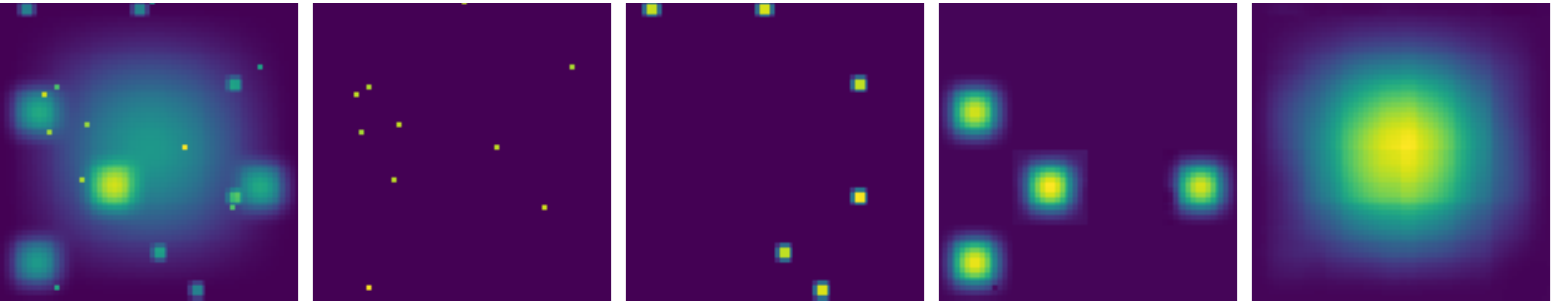}
  \caption{Multiscale low rank matrix decomposition experiment. The matrix size is $64\times64$. From left to right: observations; recovered $\Pscr_i$-block-wise low rank component for $i = 1,\dots,4$. All the blocks in $\Pscr_i$ have the same size $4^{i-1}\times4^{i-1}$ for $i = 1,\dots,4$.}
  \label{fig:multiscale}
\end{figure}

The multiscale low-rank matrix decomposition problem proposed by Ong and Lustig~\cite{ong2016beyond} generalizes the sparse and low-rank matrix decomposition through a block-wise low-rank structure. Let $X$ be an $m \times n$ matrix and $\Pscr$ be a partition of $X$ into multiple blocks. Then $X$ is considered to be block-wise low-rank with respect to $\Pscr$ if all the blocks are low rank. For each block $p\in\Pscr$ with size $m_p \times n_p$, let $X_p$ denote the corresponding part of the matrix $X$ and let $R_p: \Re^{m \times n} \to \Re^{m_p \times n_p}$ denote the linear operator that can extract $X_p$ from $X$, namely $R_p(X) = X_p$. The adjoint operator $R_p^*: \Re^{m_p \times n_p} \to \Re^{m \times n}$ embeds an $m_p \times n_p$ matrix into a $m \times n$ zero matrix. With this operator, 
\begin{equation*}
  X = \sum\limits_{p \in \Pscr} R_p^*(X_p).
\end{equation*}

Each block-wise low-rank signal is represented by a corresponding atomic set. By definition, each block $X_p \in \Re^{m_p \times n_p}$ is low rank, and thus $X_p$ is $\Ascr_p$-sparse, where 
\[
  \Ascr_p = \left\{uv^\intercal \mid u \in \Re^{m_p}, v \in \Re^{n_p}, \|u\| = \|v\| = 1\right\}.
\]
One and Lustig~\cite{ong2016beyond} propose a block-wise nuclear norm and its associated dual norm, respectively, by the functions 
\[
  \|\cdot\|_{\Pscr,1} = \textstyle \sum_{p \in \Pscr} \|R_p(\cdot)\|_1,
  \quad
  \|\cdot\|_{\Pscr,\infty} = \max_{p \in \Pscr} \|R_p(\cdot)\|_{\infty},
\]
where $\|\cdot\|_1$ and $\|\cdot\|_\infty$ are the Schatten 1- and $\infty$-norms of their matrix arguments. It follows that the block-wise norm $\|\cdot\|_{\Pscr,1}$ and dual norm $\|\cdot\|_{\Pscr,\infty}$ are the gauge and support functions, respectively, for the atomic set $\Ascr_\Pscr \coloneqq  \bigcup_{p \in \Pscr} R_p^*\Ascr_p$.

We reproduce the synthetic model described by Ong and Lustig, who construct the superposition $B = \sum_{i = 1}^k X\nai$,
where $X\nai \in \Re^{m \times n}$ is block-wise low rank with respect to the multiscale partitions $\set{\Pscr_i}_{i = 1}^k$. In our experiment, we set $m = n = 64$, $k = 4$, and for each $i\in\irange1k$,
\[m_p = n_p = 4^{i-1}  \quad \forall p \in \Pscr_i.\]
At the lowest scale $i=1$, a block-wise low-rank matrix is a scalar, and so 1-sparse matrices are included with the atomic set $\Ascr_{\Pscr_1}$.  The solutions of the deconvolution procedure \cref{eq:deconvolution} are shown in~\cref{fig:multiscale}.

\section{Looking ahead}\label{sec:8}

The random rotation model is a useful mechanism for introducing incoherence among the individual signals. However, even in contexts where it's possible to rotate the signals, it may prove too costly to do so in practice because the rotations need to be applied at each iteration of the algorithm in \cref{alg-level-set-cg}. We might then consider other mechanisms for introducing incoherence that are computationally cheaper, and rely instead, for example, on some fast random transform. The literature on demixing abounds with various incoherence notions. We wish to explore what is the relationship between these and our definition of $\beta$-incoherence. Alternative incoherence definitions may prove useful in deriving other mechanisms for inducing incoherence in the signals.

\change{%
A significant assumption of our analysis is that the parameters $\lambda_i$ exactly equilibrate the gauge values for each signal; cf.~\eqref{eq-lambda-def}. Analogous assumptions appear in many other related signal-demixing approaches \cite{mccoy2014convexity,mccoy2014sharp,oymak2017universality,mccoy2013achievable}. For example, McCoy and Tropp~\cite{mccoy2013achievable}, who also deal with the general case of recovering two or more signals, require the gauge values for each signal; cf.~\cref{eq:achievable-model}.  There are important practical cases where the parameters $\lambda_i$ are known, such as some secure communication problems, where the sender can normalize the signals before they are mixed~\cite[Section~1.3.1]{mccoy2014sharp}. In cases where parameters $\lambda_i$ are not known, however, these may be reasonably approximated by a grid search. An open research question is to analyze how the stability of the signal-recovery process depends on errors that might exist in the ideal parameter choices.%
}

\appendices
\section{Proofs} \label{sec:7}
This section contains proofs for the mathematical statements in \cref{sec:3} and \cref{sec-incoherence}. We begin with several technical results needed for analysis, which describe useful properties of descent cones. Some of these results contain their own intrinsic interest.

\subsection{Lemmas} \label{sec:lemma}

\begin{lemma}[Properties of descent cones]\label{prop-descent-cone-properties} Let $\Ascr$, $\Ascr_1$, $\Ascr_2$ be compact sets in $\Re^n$ that contain the origin in their interiors. Fix the vectors $x, x_1, x_2$. The following properties hold.
   \begin{lemmaenum}[left=18pt]
    \item \label{prop-descent-cone-properties-a}
      A vector $d$ is contained in $\Dscr(\Ascr, x)$ if and only if there is some $\bar\alpha > 0$ such that $\gauge\As(x + \alpha d) \leq \gauge\As(x)$ for all $\alpha \in [0, \bar\alpha]$;
    \item \label{prop-descent-cone-properties-b}
      $\Dscr(\tau\Ascr, x) = \Dscr(\Ascr, x)\ \forall\tau > 0$;
    \item \label{prop-descent-cone-properties-c}
      $\Dscr(Q\Ascr, Qx) = Q\Dscr(\Ascr, x)$ if $Q\in \SO(n)$;
    \item \label{prop-descent-cone-properties-d}
      $\Dscr(\Ascr_1 + \Ascr_2, x_1 + x_2) \subseteq \Dscr(\Ascr_1, x_1) + \Dscr(\Ascr_2, x_2)$ if $\gauge\Aso(x_1) = \gauge\Ast(x_2)$.
    \end{lemmaenum}
\end{lemma} 

\begin{proof}
\leavevmode
  \begin{itemize}[left=18pt]
    \item[a)] See \cite[Proposition 2.5]{mccoy2014sharp};
    \item[b)] It follows from the fact that a gauge function is positive homogenous.  
    \item[c)] Because $\gauge_{Q\Ascr}=\gauge_\Ascr(Q^*\cdot)$,
    \begin{align*}
      \Dscr(Q\Ascr, Qx) &= \cone\set{d \mid \gauge_{Q\Ascr}(Qx + d) \leq \gauge_{Q\Ascr}(Qx)}
      \\&= \cone\set{d \mid \gauge_{\Ascr}(x + Q^*d) \leq \gauge_{\Ascr}(x)}
      \\&= Q\Dscr(\Ascr, x).
    \end{align*}
    \item[d)] 
    For every $d \in \Dscr(\Ascr_1 + \Ascr_2, x_1 + x_2)$, by \cref{prop-descent-cone-properties-a}, there exists $\alpha > 0$ such that  
    \[\gauge\Asum(x_1 + x_2 + \alpha d) \leq \gauge\Asum(x_1 + x_2).\] 
    Then there exists $d_1, d_2$ such that $d_1 + d_2 = \alpha d$ and 
    \[\max\set{\gauge\Aso(x_1 + d_1), \gauge\Ast(x_2 + d_2)} \leq \gauge\Asum(x_1 + x_2).\]
    By the fact that $\gauge\Asum(x_1 + x_2) \leq \max\set{\gauge\Aso(x_1), \gauge\Ast(x_2)}$ and the assumption $\gauge\Aso(x_1) = \gauge\Ast(x_2)$, it follows that $d_i \in \Dscr(\Ascr_i, x_i)$, which implies $\alpha d = d_1 + d_2 \in \Dscr(\Ascr_1, x_1) + \Dscr(\Ascr_2, x_2)$. Thus $d \in \Dscr(\Ascr_1, x_1) + \Dscr(\Ascr_2, x_2)$.
  \end{itemize}
\end{proof}

The Gaussian width of a set $T \subset \Re^n$ is defined as
\[\omega(T) = \mE_{g} \sup\set{ \ip{g}{y} | y\in T},\]
where the expectation is taken with respect to the standard Gaussian $\Nscr(0,I_n)$. 
The following lemma summarizes the main properties that we use regarding the relationship between the conic summaries $\delta$ and $\omega$.
\begin{lemma}[Properties of conic statistical summaries]\label{prop:stat_dim}
  Let $\Kscr$ be a closed and convex cones in $\Re^n$ and let $Q\in\SO(n)$. Then the following properties hold. 
  \begin{lemmaenum}[left=18pt] 
    \item \label{prop:stat_dim_a} $\delta(Q\Kscr) = \delta(\Kscr)$;
    \item \label{prop:stat_dim_b} $\delta(\Kscr) = \mE_g \left[\sup\set{ \ip{g}{y} | y\in\Kscr\cap\mB^n}^2\right]$;
    \item \label{prop:stat_dim_c} $\omega(\Kscr\cap\mB^n)^2 \leq \delta(\Kscr)$.
  \end{lemmaenum}
\end{lemma}
\begin{proof} 
See \cite[Proposition~3.1(6) and Proposition~3.1(5)]{amelunxen2014living}, respectively, for (a) and (b). 
\begin{itemize}[left=18pt]  
  \item[c)] Indeed, we know that,
  \begin{align*}
    \omega(\Kscr\cap\mB^n)^2 &= \left[\mE_{g} \sup\set{ \ip{g}{y} | y\in \Kscr\cap\mB^n}\right]^2
    \\&\leq \mE_{g} \left[\sup\set{ \ip{g}{y} | y\in \Kscr\cap\mB^n}^2\right]
    \\&= \delta(\Kscr),
  \end{align*}
  where the first equality follows from the definition of gaussian with, the first inequality follows from the fact that $\mE(X)^2 \leq \mE(X)^2$ for any random variable $X$, and the last equality follows from \cref{prop:stat_dim_b}. 
\end{itemize}
\end{proof}

Our next lemma shows that if the angle between two cones is bounded, then the norms of individual vectors are bounded by the norm of their sum. 
\begin{lemma}\label{lemma:bound_norm}
  Let $\Kscr_1$ and $\Kscr_2$ be two closed convex cones in $\Re^n$. If $\cos\angle(-\Kscr_1, \Kscr_2) \leq 1 - \beta$ for some $\beta \in (0, 1]$, then for any $u \in \Kscr_1$ and $v \in \Kscr_2$, 
  \[\max\set{\|u\|, \|v\|} \leq \frac{1}{\sqrt{\beta}}\|u + v\|.\]
\end{lemma}

\begin{proof}
By expanding the norm square of $u + v$ we can get that 
  \begin{align*}
    \|u + v\|^2 &= \|u\|^2 + \|v\|^2 - 2\ip{-u}{v}
    \\&= \|u\|^2 + \|v\|^2 - 2\cos(\angle(-u, v))\|u\|\|v\|
    \\&\geq \|u\|^2 + \|v\|^2 - 2(1-\beta)\|u\|\|v\|
    \\&= \beta(\|u\|^2 + \|v\|^2) + (1-\beta)(\|u\| - \|v\|)^2
    \\&\geq \beta\max\set{\|u\|^2, \|v\|^2},
  \end{align*}
where the first inequality follows from the definition of the cosine of the angle between two cones. 
\end{proof}

Our next lemma is a technical lemma for the expectation. 
\begin{lemma}\label{lemma:expectation}
  Let $X$ and $Y$ be nonnegative random variables, then we have 
  \[\mE[(X+Y)^2] \leq \left( \sqrt{\mE[X^2]} + \sqrt{\mE[Y^2]} \right)^2.\]
\end{lemma}

\begin{proof}
  By expanding the right hand side, we can get
  \begin{align*}
    \left( \sqrt{\mE[X^2]} + \sqrt{\mE[Y^2]} \right)^2 &= \mE[X^2] + \mE[Y^2] + 2\sqrt{\mE[X^2]\mE[Y^2]}
    \\&\geq \mE[X^2] + \mE[Y^2] + 2\mE[XY]
    \\&= \mE[(X+Y)^2],
  \end{align*}
  where the inequality follows from the Cauchy–Schwarz inequality.
\end{proof}

\subsection{Proof for \titlelink{\cref{thm-stability}}} \label{sec:proof-thm-stability}
For each $i\in\irange1k$, let $\epsilon_i \coloneqq  x_i^* - x\nai$ and $\epsilon_{-i} : = \sum_{j \neq i}\epsilon_j$. By the definition of descent cone,
$\epsilon_i \in \Dscr(\Ascr_i, x\nai$).
Because $\{(x\nai, \Ascr_i)\}_{i=1}^k$ are $\beta$-incoherent for some $\beta\in(0,1]$, by~\cref{def:incoherence},
\[\cos\angle\left(-\epsilon_i, \epsilon_{-i}\right) \leq 1 - \beta.\]
By \cref{lemma:bound_norm}, it follows that
\[ \left\|\epsilon_i + \epsilon_{-i}\right\| \geq \sqrt{\beta}\|\epsilon_i\|.\]
The desired result follows.

\subsection{Proof for \titlelink{\cref{prop:bound_sta_dim}}} \label{sec:proof-bound_sta_dim}
In this proof, we define $\Kscr_{i, \beta} =  \Kscr_i \cap \tfrac{1}{\sqrt{\beta}} \mB^n$ and $f_{i, \beta}(g) = \sup\set{ \ip{g}{u} | u \in \Kscr_{i, \beta}}$ for $i = 1, 2$. By \cref{prop:stat_dim_b}, we know that the statistical dimension can be expressed as
\begin{align*}
  &\delta(\Kscr_1 + \Kscr_2) = \mE_g \left[\sup\set{ \ip{g}{y} | y\in(\Kscr_1 + \Kscr_2)\cap\mB^n}^2 \right]
  \\&= \mE_g \left[\sup\set{ \ip{g}{u + v} | u \in \Kscr_1,\ v \in \Kscr_2, \|u + v\| \leq 1 }^2\right]
  \\&\leq \mE_g \left[\sup\set{ \ip{g}{u + v} | u \in \Kscr_{1, \beta},\ v \in \Kscr_{2, \beta} }^2\right]
  \\&= \mE_g \left[ \left(\sup_{u \in \Kscr_{1, \beta}}\ \ip{g}{u} + \sup_{v \in \Kscr_{2, \beta}}\ \ip{g}{v} \right)^2 \right]
  \\&\leq \left( \sqrt{ \mE_g \left[ f_{1, \beta}(g)^2 \right]} + \sqrt{ \mE_g \left[ f_{2, \beta}(g)^2 \right]} \right)^2
  \\&= \tfrac{1}{\beta}\left(\sqrt{\delta(\Kscr_1)} + \sqrt{\delta(\Kscr_2)} \right)^2,
\end{align*}
where the first inequality follows from \cref{lemma:bound_norm} and the fact that the supremum is always nonnegative, and the second inequality follows from \cref{lemma:expectation}. 

\subsection{Proof for \titlelink{\cref{coro:bound_sta_dim}}} \label{sec:proof-coro-bound_sta_dim}
Throught this proof, for all $i\in1:k$, we define $\Dscr_i = \Dscr(\Ascr_i, x\nai)$, $\delta_i = \delta(\Dscr_i)$ and $\delta_{1:i} = \delta(\sum_{j=1}^i\Dscr_i)$. By \cref{assume-blanket} and \cref{prop-descent-cone-properties-d}, we know that $\Dscr\subs \subseteq \sum_{i=1}^k \Dscr_i$, and it follows that $\delta(\Dscr\subs)\leq\delta_{1:k}$. So we only need to give an upper bound for $\delta_{1:k}$. Since $\cos\angle\left({-\Dscr_k},\, \sum_{i=1}^{k-1}\Dscr_i\right)\leq 1 - \beta$, by \cref{prop:bound_sta_dim}, it follows that 
\begin{equation} \label{eq:11}
  \sqrt{\delta_{1:k}} \leq \beta^{-\tfrac{1}{2}}\left(\sqrt{\delta_{1:(k-1)}} + \sqrt{\delta_k}\right).
\end{equation}
Since $\sum_{i=1}^{k-2}\Dscr_i \subseteq \sum_{j\neq(k-1)}\Dscr_j$, it follows that $\cos\angle\left({-\Dscr_{k-1}},\, \sum_{i=1}^{k-2}\Dscr_i\right)\leq  1 - \beta$. By \cref{prop:bound_sta_dim}, we have 
\begin{equation} \label{eq:12}
  \sqrt{\delta_{1:(k-1)}} \leq \beta^{-\tfrac{1}{2}}\left(\sqrt{\delta_{1:(k-2)}} + \sqrt{\delta_{k-1}}\right).
\end{equation}
Combining \eqref{eq:11} and \eqref{eq:12}, we know that 
\[\sqrt{\delta_{1:k}} \leq \beta^{-\tfrac{2}{2}}\left(\sqrt{\delta_{1:(k-2)}} + \sqrt{\delta_{k-1}} + \sqrt{\delta_k}\right).\]
Repeating this process, we can conclude that 
\[\sqrt{\delta_{1:k}} \leq \beta^{-\tfrac{k-1}{2}}\sum_{i=1}^k\sqrt{\delta_i}.\]


\subsection{Proof for \titlelink{\cref{prop-angle-cones}}}\label{sec:proof-prop-angle-cones}
Throught this proof, we define the following notations:
\begin{itemize} 
  \item $\overline\Kscr_i:=\Kscr_i\cap\mS^{n-1}$ for $i=1,2$;
  \item $\widehat\Kscr_i:=\Kscr_i\cap\mB^n$ for $i=1,2$;
  \item $f(W\in\Re^{n\times n}) = \sup\set{\ip{x}{Wy} \mid x \in \overline\Kscr_1, y \in\overline\Kscr_2}$;
  \item $\hat f(W\in\Re^{n\times n}) = \sup\set{\ip{x}{Wy} \mid x \in \widehat\Kscr_1, y \in\widehat\Kscr_2}$;
  \item $\Oscr_n = \set{Q\in\Re^{n\times n}: Q^TQ = I_n}$;
  \item $\Sscr\Oscr_{n, +} = \set{Q\in\Oscr_n: \det(Q) = 1}$;
  \item $\Sscr\Oscr_{n, -} = \set{Q\in\Oscr_n: \det(Q) = -1}$.
\end{itemize} 

Our proof consists of three steps. 

\textbf{First step: show that both $f$ and $\hat f$ are convex and $1$-Lipschitz functions.} First, we show that both $f$ and $\hat f$ are convex. 
For any $W_1, W_2 \in \Re^{n\times n}$ and any $t \in [0,1]$, 
\begin{align*}
  &f(tW_1 + (1-t)W_2) 
  \\&= \sup\set{\ip{x}{(tW_1 + (1-t)W_2)y} \mid x \in \overline\Kscr_1, y \in\overline\Kscr_2}
  \\&= \sup\set{\ip{x}{tW_1y} + \ip{x}{(1-t)W_2y} \mid x \in \overline\Kscr_1, y \in\overline\Kscr_2}
  \\&\leq tf(W_1) + (1-t)f(W_2).
\end{align*}
So $f$ is convex. The same reason holds for $\hat f$, and thus $\hat f$ is also convex. 
Next, by \cite[Lemma~2.6]{shalev2011online}, in order to show that both $f$ and $\hat f$ are $1$-Lipschitz, we only need to show that the norm of any subgradient of $f$ or $\hat f$ is bounded by $1$. By \cite[Theorem~D.4.4.2]{hiriart-urruty01}, we know that for any $W \in \Re^{n\times n}$, 
\begin{align*}
  \partial f(W) &= \conv\set{xy^T \mid x \in \overline\Kscr_1, y \in\overline\Kscr_2, \ip{x}{Wy} = f(W)},
  \\ \partial \hat f(W) &= \conv\set{xy^T \mid x \in \widehat\Kscr_1, y \in\widehat\Kscr_2, \ip{x}{Wy} = f(W)}.
\end{align*}
Since $\|x\| \leq 1$ and $\|y\| \leq 1$, it is easy to verify that for any $W\in\Re^{n\times n}$ and for any $Z \in \partial f(W) \cup \partial \hat f(W)$, 
\[\|Z\|_F \leq 1,\]
where $\|\cdot\|_F$ is the Frobenius norm. Therefore, we can conclude that both $f$ and $\hat f$ are $1$-Lipschitz functions.

\textbf{Second step: bound $\mE_{Q \sim \Uscr(\Sscr\Oscr_{n, +})}\left[f(Q)\right]$.} 
First, we give the bound on $\mE_{Q \sim \Uscr(\Oscr_n)}\left[\hat f(Q)\right]$. From the first step, we know that $\hat f$ is convex. Then by the comparison principle developed by Tropp; see \cite[Theorem~5 and Lemma~8]{tropp2012comparison}, we can conclude that 
\begin{equation} \label{eq:help1}
  \mE_{Q \sim \Uscr(\Oscr_n)}\left[\hat f(Q)\right] \leq \tfrac{1.5}{\sqrt{n}}\mE_{G\sim\Nscr(0, I_n)}\left[\hat f(G)\right],
\end{equation}
Next, we give the bound on $\mE_{Q \sim \Uscr(\Sscr\Oscr_{n, +})}\left[\hat f(Q)\right]$. By expanding the uniform distribution over $\Oscr_n$, we can get
\begin{align*}
  &\mE_{Q \sim \Uscr(\Oscr_n)}\left[\hat f(Q)\right] 
  \\&= \tfrac{1}{2}\mE_{Q \sim \Uscr(\Sscr\Oscr_{n, +})}\left[\hat f(Q)\right] + \tfrac{1}{2}\mE_{Q \sim \Uscr(\Sscr\Oscr_{n, -})}\left[\hat f(Q)\right]
  \\&\geq \tfrac{1}{2}\mE_{Q \sim \Uscr(\Sscr\Oscr_{n, +})}\left[\hat f(Q)\right],
\end{align*}
where the inequality follows from the fact that $\hat f$ is non-negative. Combine this result with \eqref{eq:help1}, we can conclude that 
\begin{equation} \label{eq:help2}
  \mE_{Q \sim \Uscr(\Sscr\Oscr_{n, +})}\left[\hat f(Q)\right] \leq \tfrac{3}{\sqrt{n}}\mE_{G\sim\Nscr(0, I_n)}\left[\hat f(G)\right].
\end{equation}
Then, by the Gaussian Chevet’s inequality; see\cite[Exercise~8.7.4]{vershynin2018high}, we know that 
\begin{equation} \label{eq:help3}
\begin{split}
  \mE_{G\sim\Nscr(0, I_n)}\left[\hat f(G)\right] &\leq \omega(\widehat\Kscr_1) + \omega(\widehat\Kscr_2)\\
  &\leq \sqrt{\delta(\Kscr_1)} + \sqrt{\delta(\Kscr_2)},
\end{split}
\end{equation}
where the second inequality follows from \cref{prop:stat_dim_c}. Combine \eqref{eq:help2} and \eqref{eq:help3}, we can get
\begin{equation} \label{eq:help4}
  \mE_{Q \sim \Uscr(\Sscr\Oscr_{n, +})}\left[\hat f(Q)\right] \leq \tfrac{3}{\sqrt{n}}\left( \sqrt{\delta(\Kscr_1)} + \sqrt{\delta(\Kscr_2)} \right).
\end{equation}
Finally, by the fact that $f \leq \hat f$ and \eqref{eq:help4}, we can conclude that 
\begin{equation} \label{eq:help5}
  \mE_{Q \sim \Uscr(\Sscr\Oscr_{n, +})}\left[f(Q)\right] \leq \tfrac{3}{\sqrt{n}}\left( \sqrt{\delta(\Kscr_1)} + \sqrt{\delta(\Kscr_2)} \right).
\end{equation}

\textbf{Third step: concentration bound for $f(Q)$.}
From step 1, we know that $f$ is $1$-Lipschitz. For clearness, we denote $\mP_{Q \sim \Uscr(\Sscr\Oscr_{n, +})}$ and $\mE_{Q \sim \Uscr(\Sscr\Oscr_{n, +})}$ as $\mP_Q$ and $\mE_Q$. By the concentration bounds of Lipschitz functions over the special orthogonal group develop by Meckes; see \cite[Theorem~5.5 and Theorem~5.16]{meckes2019random}, we can get that for every $t\geq0$,
\[\mP_Q\left[f(Q) \geq \mE_Q[f(Q)] + t\right] \leq \exp(-\tfrac{n-2}{8}t^2).\]
Note that a similar result can be obtained from \cite[Theorem~5.2.7]{vershynin2018high}.
Combining with \eqref{eq:help5}, we can conclude that for every $t\geq0$,
\[\mP_Q\left[f(Q) \geq  \tfrac{3}{\sqrt{n}}\left( \sqrt{\delta(\Kscr_1)} + \sqrt{\delta(\Kscr_2)} \right) + t\right] \leq \exp(-\tfrac{n-2}{8}t^2).\]

\subsection{Lemmas needed for the proof of \titlelink{\cref{thm:Incoherence}}} \label{sec:prob_bounds}
In this section, we present two lemmas that are needed for the proof of \cref{thm:Incoherence}. These two lemmas provide probabilistic bound on the statistical dimension of sum of randomly rotated cones. 

The next lemma provides a probabilistic bound on the statistical dimension of the sum of two cones. 
\begin{lemma}[Probabilistic bound on statistical dimension under random rotation] \label{prop-bound-sta-dim}
   Let $\Kscr_1$ and $\Kscr_2$ be two closed convex cones in $\Re^n$. Then 
 \begin{align*}
   \mP\bigg[\sqrt{\delta(\Kscr_1 + Q\Kscr_2)} &\leq \tfrac{1}{\sqrt{\beta(t)}}\left(\sqrt{\delta(\Kscr_1)} + \sqrt{\delta(\Kscr_2)} \right) \bigg]
   \\&\geq 1 - \exp(-\tfrac{n-2}{8}t^2)
   \\\text{with} \beta(t) &= 1 - \tfrac{3}{\sqrt{n}}\left(\sqrt{\delta(\Kscr_1)} + \sqrt{\delta(\Kscr_2)}\right) - t ,
 \end{align*}
   where $Q$ is drawn uniformly at random from $\SO(n)$.
\end{lemma}

\begin{proof}
  By \cref{prop:bound_sta_dim} and \cref{prop-angle-cones}.
\end{proof}

Our next lemma extends \cref{prop-bound-sta-dim} to arbitrary number of cones. 
\begin{lemma} \label{prop:bound-sum-sta-dim}
   Let $\Kscr_1, \dots, \Kscr_p$ be closed convex cones in $\Re^n$ and let $Q_1, \dots, Q_p$ be i.i.d. matrices uniformly drawn from $\SO(n)$. If $\sum_{i=1}^p \sqrt{\delta(\Kscr_i)} \leq \left(1 - 4^{- \scaleto{\tfrac{1}{p-1}}{10pt} } - t\right)\sqrt{n} / 6$ for some $t > 0$, then 
   \begin{equation*}
   \mP\left[ \sqrt{\delta\left(\overline{\Kscr}\right)} \leq 2\sum_{i=1}^p \sqrt{\delta(\Kscr_i)}\right] \geq 1 - (p-1)\exp(-\tfrac{n-2}{8}t^2),
   \end{equation*}
   where $\overline{\Kscr} = \sum_{i=1}^pQ_i\Kscr_i$.
\end{lemma}

\begin{proof}
Throughout this proof, we define the following notations:
\begin{itemize}
  \item $\delta_i = \delta(\Kscr_i)$, for all $i\in1:p$;
  \item $\delta_{1:i} = \delta\left(\sum_{j=1}^iQ_j\Kscr_j\right)$, for all $i\in1:p$;
  \item For each $i\in2:p$, define the event 
\begin{align*}
  E_i(t) &= \left\{ \sqrt{\delta_{1:i}} \leq \tfrac{1}{\sqrt{\beta_i(t)}} \left( \sqrt{\delta_{1:(i-1)}} + \sqrt{\delta_i}\right)\right\}
  \\\text{with} \beta_i(t) &= 1 - \tfrac{3}{\sqrt{n}}\left( \sqrt{\delta_{1:(i-1)}} + \sqrt{\delta_i}\right) - t.
\end{align*}
\end{itemize}

\textbf{Our proof consists of three steps.}

\textbf{Step 1: bound the probability of $E_2(t)\land\cdots\land E_p(t)$.}
Denote the indicator random variable for $E_i(t)$ by $\mathbbm{1}_{E_i(t)}$, which evaluates to $1$ if $E_i(t)$ occurs and otherwise evaluates to $0$. Then for each $i\in2:p$, we have
\begin{align*}
  \mP(E_i(t)) &= \mE(\mathbbm{1}_{E_i(t)})
  \\&= \mE_{\{Q_j\}_{j=1}^{i-1}} \left[ \mE\left(\mathbbm{1}_{E_i(t)} \mid \{Q_j\}_{j=1}^{i-1}\right)\right]
  \\&\geq \mE_{\{Q_j\}_{j=1}^{i-1}} \left[ 1 - \exp(-\tfrac{n-2}{8}t^2) \right]
  \\&= 1 - \exp(-\tfrac{n-2}{8}t^2),
\end{align*}
where the inequality follows from \cref{prop-bound-sta-dim} and the assumption that $Q_i$ are all independent. Extending the bound on $\mP(E_i(t))$ to all $i\in\irange2p$ via the union bound, we have
\[
  \mP(E_2(t) \land\cdots\land E_p(t))
  \geq 1 - (p-1)\exp(-\tfrac{n-2}{8}t^2).\]

\textbf{Step 2: show that $E_2(t) \land\cdots\land E_p(t)$ implies bound on $\sqrt{\delta_{1:p}}\leq\tfrac{1}{\sqrt{\beta_2(t)\dots\beta_p(t)}}\sum_{i=1}^p\sqrt{\delta_i}$}. Indeed, we have 
\begin{align*}
  \sqrt{\delta_{1:p}} &\leq \tfrac{1}{\sqrt{\beta_p(t)}}\left( \sqrt{\delta_{1:(p-1)}} + \sqrt{\delta_p}\right)
  \\&\leq \tfrac{1}{\sqrt{\beta_p(t)}}\left( \tfrac{1}{\sqrt{\beta_{p-1}(t)}}\left( \sqrt{\delta_{1:(p-2)}} + \sqrt{\delta_{p-1}}\right) + \sqrt{\delta_p}\right)
  \\&\leq \tfrac{1}{\sqrt{\beta_p(t)\beta_{p-1}(t)}}\left( \sqrt{\delta_{1:(p-2)}} + \sqrt{\delta_{p-1}} + \sqrt{\delta_p}\right)
  \\&\vdots
  \\&\leq \tfrac{1}{\sqrt{\beta_2(t)\dots\beta_i(t)}}\sum_{j=1}^i\sqrt{\delta_j}.
\end{align*}

\textbf{Step 3: show that $E_2(t) \land\cdots\land E_p(t)$ and the assumption $\sum_{i=1}^p \sqrt{\delta(\Kscr_i)} \leq \left(1 - 4^{- \scaleto{\tfrac{1}{p-1}}{10pt} } - t\right)\sqrt{n} / 6$ implies that $\beta_i(t) \geq 4^{-\scaleto{\tfrac{1}{k-1}}{10pt} }$ for $i\in2:p$.} We prove this by induction on $i$. First we show that $\beta_2(t) \geq 4^{-\scaleto{\tfrac{1}{k-1}}{10pt} }$. Indeed, we have 
\begin{align*}
  \beta_2(t) &= 1 - \tfrac{3}{\sqrt{n}}\left( \sqrt{\delta_1} + \sqrt{\delta_2}\right) - t
  \\&\geq 1 - \tfrac{3}{\sqrt{n}}\tfrac{ \left(1 - 4^{- \scaleto{\tfrac{1}{k-1}}{10pt} } - t\right)\sqrt{n} }{6} - t
  \geq  4^{-\scaleto{\tfrac{1}{k-1}}{10pt} }.
\end{align*}
 Next for any $i\in3:k$, we assume that $\beta_j(t) \geq 4^{-\scaleto{\tfrac{1}{k-1}}{10pt} }$ for all $2\leq j \leq (i-1)$, then we have 
\begin{align*}
  \beta_i(t) &=  1 - \tfrac{3}{\sqrt{n}}\left( \sqrt{\delta_{1:(i-1)}} + \sqrt{\delta_i}\right) - t
  \\&\geq 1 - \tfrac{3}{\sqrt{n}} \tfrac{1}{\sqrt{\beta_2(t)\dots\beta_{i-1}(t)}}\sum_{j=1}^{i}\sqrt{\delta_j} - t
  \\&\geq 1 - \tfrac{3}{\sqrt{n}} 2^{ \scaleto{\tfrac{i-2}{k-1}}{10pt} } \tfrac{ \left(1 - 4^{- \scaleto{\tfrac{1}{k-1}}{10pt} } - t\right)\sqrt{n} }{6} - t
  \geq 4^{- \scaleto{\tfrac{1}{k-1}}{10pt} }.
\end{align*}

Finally, combining all three steps, we can conclude that 
\begin{equation*}
  \mP\left[ \sqrt{\delta\left(\overline{\Kscr}\right)} \leq 2\sum_{i=1}^p \sqrt{\delta(\Kscr_i)}\right] \geq 1 - (p-1)\exp(-\tfrac{n-2}{8}t^2),
\end{equation*}
with $\overline{\Kscr} = \sum_{i=1}^pQ_i\Kscr_i$.
\end{proof}

\subsection{Proof for \titlelink{\cref{thm:Incoherence}}} \label{sec:proof-thm:Incoherence}
Throughout this proof, we define the following notations for all $i\in1:k$:
\begin{itemize} 
  \item $\Dscr_i = \Dscr(\Ascr_i, x\nai)$;
  \item $\hat \Dscr_i =  \Dscr(\hat \Ascr_i, \hat x\nai)$;
  \item $\delta_i = \delta(\Dscr_i)$;
  \item $\delta_{1:i} = \delta\left(\sum_{j=1}^i\Dscr_i\right)$;
  \item $\delta_{-i} = \delta\left(\sum_{j\neq i}\Dscr_j\right)$.
\end{itemize} 
By~\cref{prop-descent-cone-properties-c}, for all $i\in1:k$, we have
\[\Dscr_i = \Dscr(Q_i \hat \Ascr_i, Q_i \hat x\nai) = Q_i\hat\Dscr_i.\]
Then it follows from \cref{prop:stat_dim_a} that
\[\delta(\hat \Dscr_i) = \delta(Q_i^T \Dscr_i) = \delta_i.\]
For all $i\in\irange1k$, define 
\begin{itemize}
  \item $\hat \Dscr_i\coloneqq  \Dscr(\hat \Ascr_i, \hat x\nai)$;
  \item $\delta_i = \delta(\Dscr_i)$;
  \item $\delta_{-i} = \delta\left(\sum_{j\neq i}\Dscr_j\right)$
\end{itemize}
First, fix $t >0$, for each $i\in\irange1k$, define the event
\[
E_i(t) =  \left \{ \cos\angle\left(-\Dscr_i, \sum_{j \neq i}\Dscr_j\right) \leq \tfrac{3}{\sqrt{n}}\left(\sqrt{\delta_i} + \sqrt{\delta_{-i}}\right) + t \right\}.
\] 
Denote the indicator random variable for $E_i(t)$ by $\mathbbm{1}_{E_i(t)}$, which evaluates to $1$ if $E_i(t)$ occurs and otherwise evaluates to $0$. Then, the following chain of inequalities gives the upper bound for the probability of the event $E_i(t)$:
\begin{align} \label{eq:1}
  \begin{split}
\mP (E_i(t)) 
  &= \mE (\mathbbm{1}_{E_i(t)})
\\&= \mE_{\{Q_j\}_{j\ne i}} \mE\left[\mathbbm{1}_{E_i(t)} \bigm\vert Q_j\ \forall j\ne i \right]
\\&\geq \mE_{\{Q_j\}_{j\ne i}} [ 1 - \exp(-\tfrac{n-2}{8}t^2) ]
\\&= 1 - \exp(-\tfrac{n-2}{8}t^2),
\end{split}
\end{align}
where the inequality follows from \cref{prop-angle-cones}. 

Next, by \cref{prop:bound-sum-sta-dim}, we know that 
\begin{equation}\label{eq:2}
  \mP\left[ \sqrt{\delta_{-i}} \leq 2\sum_{j\neq i} \sqrt{\delta_j}\right] \geq 1 - (k-2)\exp(-\tfrac{n-2}{8}t^2).
\end{equation}

Thirdly, for each $i\in\irange1k$, define the event, 
\[
\hat E_i(t) =  \left \{ \cos\angle\left(-\Dscr_i, \sum_{j \neq i}\Dscr_j\right) \leq \tfrac{6}{\sqrt{n}}\sum_{i=1}^k\sqrt{\delta_i} + t\right \}.
\]
By combining \eqref{eq:1} and \eqref{eq:2} together, we can conclude that 
\[\mP(\hat E_i(t)) \geq 1 - (k-1)\exp(-\tfrac{n-2}{8}t^2).\]
Extend the bound on $\mP(\hat E_i(t))$ to all $i\in\irange1k$ via the union bound:
\[\mP(\hat E_1\land\cdots\land \hat E_k) \geq 1 - k(k-1)\exp(-\tfrac{n-2}{8}t^2).\]

Finally, by our assumption that $\sum_{i=1}^k \sqrt{\delta(\Dscr_i)} \leq \left(1 - 4^{- \scaleto{\tfrac{1}{k-1}}{10pt} } - t\right)\sqrt{n} / 6$, it follows that 
\[\tfrac{6}{\sqrt{n}}\sum_{i=1}^k\sqrt{\delta_i} + t \leq 1 - 4^{- \scaleto{\tfrac{1}{k-1}}{10pt} }.\]
Therefore, we can conclude that the rotated pairs $\{(x\nai,\Ascr_i)\}_{i=1}^k$ are $4^{- \scaleto{\tfrac{1}{k-1}}{10pt} }$-incoherent with probability at least $1 - k(k-1)\exp(-\tfrac{n-2}{8}t^2)$.

\ifCLASSOPTIONcaptionsoff
  \newpage
\fi

\bibliographystyle{IEEEtran}
\bibliography{refs/shorttitles, refs/master, refs/friedlander}

\end{document}

%% file: macro/local-macros.tex
\usepackage{relsize}
\newcommand{\polar}{^\circ}

\newcommand{\gauge}{\gamma}
\newcommand{\As}{_{\scriptscriptstyle\Ascr}}

\newcommand{\Aso}{_{\scriptscriptstyle\Ascr_1}}
\newcommand{\Ast}{_{\scriptscriptstyle\Ascr_2}}
\newcommand{\Asi}{_{\scriptscriptstyle\Ascr_i}}

\newcommand{\Asum}{_{\scriptscriptstyle\Ascr_1 + \scriptscriptstyle\Ascr_2}}

\newcommand{\Ass}{_{\scriptscriptstyle\Ascr_s}}

\newcommand{\na}{^{\natural}}
\newcommand{\nao}{_1^{\natural}}
\newcommand{\nat}{_2^{\natural}}
\newcommand{\nai}{_i^{\natural}}

\newcommand{\nag}{_{\scriptscriptstyle S}^{\natural}}

\DeclareMathOperator{\SO}{SO}
\DeclareMathOperator{\mE}{\mathop{\mathbb{E}}}
\DeclareMathOperator{\mP}{\mathop{\mathbb{P}}}
\newcommand{\maxconv}{\mathop{\diamond}}
\newcommand{\titlelink}[1]{\texorpdfstring{#1}{}}
\newcommand{\xagg}{x\subs^\natural}
\newcommand{\irange}[2]{#1\mathord{:}#2}